\documentclass[11pt,reqno]{amsart}

\usepackage{mathrsfs}
\usepackage{amsmath}
\usepackage{amssymb}
\usepackage{amsfonts}
\usepackage{amsthm}
\usepackage{stmaryrd}
\usepackage{linearA}
\usepackage{float}

\usepackage{dsfont}
\usepackage[T1]{fontenc}
\usepackage[english]{babel}

\usepackage{a4wide}
\usepackage{graphicx}
\usepackage{pgfkeys}
\usepackage{longtable}
\usepackage{pdflscape}
\usepackage{color}
\usepackage{array}
\usepackage[left=2.5cm,right=2.5cm,top=3cm,bottom=3cm]{geometry}
\numberwithin{equation}{section}
\usepackage{comment}
\usepackage{tikz}
\usetikzlibrary{shapes}

\usepackage{enumitem}

\usepackage{hyperref}
\hypersetup{pdfauthor={Name}}
\hypersetup{colorlinks=true,allcolors=[rgb]{0,0,0.6}}

\newtheorem{theorem}{Theorem}[section]
\newtheorem{lemma}[theorem]{Lemma}
\newtheorem{proposition}[theorem]{Proposition}
\newtheorem{corollary}[theorem]{Corollary}
\newtheorem{definition}[theorem]{Definition}

\newtheorem{assumption}{Hypothesis}

\DeclareMathOperator*{\slim}{s-lim} 
\DeclareMathOperator*{\wlim}{w-lim} 
\DeclareMathOperator{\Ran}{Ran} 
\DeclareMathOperator{\Ker}{Ker} 
\DeclareMathOperator{\Id}{Id} 
\DeclareMathOperator{\Span}{Span} 
\DeclareMathOperator{\im}{Im} 
\DeclareMathOperator{\re}{Re} 
\newcommand{\Hi}{\mathcal{H}} 
\newcommand{\B}{\mathcal{B}}
\newcommand{\C}{\mathbb{C}} 
\newcommand{\R}{\mathbb{R}} 
\newcommand{\N}{\mathbb{N}} 
\newcommand{\norme}[1]{\left\Vert #1\right\Vert} 
\newcommand{\scal}[2]{\left\langle{#1} \middle., {#2}\right\rangle} 
\newcommand{\abso}[1]{\left|{#1}\right|}
\newcommand{\Res}{\mathcal{R}}

\author[N. Frantz]{Nicolas Frantz}
\address[N. Frantz]{Laboratoire Analyse Géométrie Modélisation \\
CY Cergy Paris Universit{\'e} 
95302 Cergy Pontoise, France}
\email{nicolas.frantz@cyu.fr}
\title{Scattering theory for some non-self-adjoint operators}
\begin{document}

\maketitle
\begin{abstract}
	We consider a non-self-adjoint $H$ given as the perturbation of a self-adjoint operator $H_0$. We suppose that $H$ is of the form $H=H_0+CWC$ where $C$ is a bounded, positive definite and relatively compact with respect to $H_0$, and $W$ is bounded. We suppose that $C(H_0-z)^{-1}C$ is uniformly bounded in $z\in\mathbb{C}\setminus\mathbb{R}$. We define the regularized wave operators associated to $H$ and $H_0$ by $W_\pm(H,H_0):=\displaystyle\slim_{t\rightarrow\infty} e^{\pm itH}r_\mp(H)\Pi_\mathrm{p}(H^\star)^\perp e^{\mp itH_0}$ where $\Pi_\mathrm{p}(H^\star)$ is the projection onto the direct sum of all the generalized eigenspace associated to eigenvalue of $H^\star$ and $r_\mp$ is a rational function that regularizes the `incoming/outgoing spectral singularities' of $H$. We prove the existence and study the properties of the regularized wave operators. In particular we show that they are asymptotically complete if $H$ does not have any spectral singularity. 
\end{abstract}
\tableofcontents
\newpage

\section{Introduction}
We are interested in this paper in the scattering theory for non-self-adjoint operators. Non-self-adjoint `Hamiltonians' in Quantum Mechanics are particularly relevant in various context, see e.g. \cite{Ba15_01,Kr17_01} and references therein for a detailed exposition.

In particular, non-conservative phenomena may be described by effective or phenomenological non-self-adjoint operators. A famous model involving non-self-adjoint operators is the \emph{nuclear optical model} describing the interaction of a neutron (or a proton) and a nucleus in Nuclear Physics. It was introduced by Feshbach, Porter and Weisskopf in \cite{FePoWe54_01} as an empirical model taking into account the formation of a compound nucleus. 
In this model, if the neutron is initially (at time $t=0$) in the normalized state $\varphi_0$, then it evolves into the unnormalized state $\varphi_t=e^{-itH}\varphi_0$ at time $t$, where $\varphi_t$ is the solution to the Schr\"odinger equation 
\begin{equation}\label{eq:Sch equ}
	i\partial_t\varphi_t=H\varphi_t.
\end{equation}
Here $H=-\Delta + V(x)$ is a dissipative operator acting on $L^2(\R^3,\C)$, with $\im(V(x))\leq 0$. Part of the energy of the neutron may be transferred to the nucleus, eventually leading to the capture of the neutron by the nucleus. The transfer of energy is mathematically illustrated by the dissipative nature of the `Hamiltonian' $H$. In particular, given a normalized state $\varphi_0$, the probability that the neutron escapes from the nucleus is given by 
\begin{equation*}
	\rho_{\mathrm{scatt}}(\varphi_0)=\lim_{t\rightarrow\infty}\norme{e^{-itH}\varphi_0}_{L^2}^2.
\end{equation*}
and $\rho_\mathrm{abs}(\varphi_0)$, the probability of absorption of the neutron by the nucleus, is given by 
\begin{equation*}
	\rho_\mathrm{abs}(\varphi_0)=1-\rho_\mathrm{scatt}(\varphi_0).
\end{equation*}
If the neutron is initially in a state whose probability of scattering is not zero, then it is expected that there exists a scattering state $\varphi\in\Hi$ such that $\norme{\varphi}_{L^2(\R^3,\C)}=\rho_\mathrm{scatt}(\varphi_0)$ and 
\begin{equation*}
	\lim_{t\rightarrow\infty}\norme{e^{-itH}\varphi_0-e^{-itH_0}\varphi}_{L^2(\R^3,\C)}=0.
\end{equation*}
This motivate the development of a scattering theory for non-self-adjoint operators.
The nuclear optical model leads to predictions that correspond to experimental scattering data to a high precision. Theoretical justifications of the model have been given in \cite{Fe58_01,Fe58_02,Fe62_01} (see also \cite{Ba15_01,Fe92_01,Ho71_01}). We mention the works \cite{FaFr18_01,FaNi18_01,Fa21_01} for an abstract dissipative scattering theory and \cite{Mo_76/77,Si79} for dissipative Schr\"odinger operators. 
Moreover, scattering theory for non-self-adjoint operators on Hilbert spaces has been considered in other contexts, see for example \cite{KaYa76_09} for the construction of local wave operators assuming that a limiting absorption principle holds, \cite{St13_11} for one-dimensional Schr\"odinger operators with a complex potential on the half-line, \cite{StS19_05} for an abstract framework on scattering theory for non-self-adjoint operators under an assumption of Kato's smoothness of the perturbation, and \cite{Ka65_01} in the case where the perturbation is not too large in Kato's sense.

We consider here an abstract class of operators of the form $H=H_0+V$ acting on a Hilbert space $\Hi$, where $H_0$ is self-adjoint with absolutely continuous spectrum and $V$ is a bounded operator relatively compact with respect to $H_0$. In particular, the essential spectrum of $H$ and that of $H_0$ coincide. We suppose that $V$ decomposes into the form $V=CWC$, with $W$ bounded and $C$ a bounded operator such that 
\begin{equation}\label{eq: Lim abs H_0}
	\sup_{z\in\C\backslash \R}\norme{C(H_0-z)^{-1}C}_{\mathcal{B}(\Hi)}<\infty.
\end{equation}
Such factorizations go back to the seminal work of Kato \cite{Ka65_01}.

We aim at constructing and studying the wave operators associated to $H$ and $H_0$. As $H$ is a perturbation of a self-adjoint operator by a bounded operator, $-iH$ is the generator of a strongly continuous group of evolution. 
We suppose in addition that $H$ (and hence $H^\star$) has a finite number of eigenvalues with finite algebraic multiplicities. We can then define the \emph{regularized wave operators} associated to $H$ and $H_0$ by
\begin{equation}\label{eq:def_WO-intro}
	W_\pm(H,H_0)=\slim_{t\rightarrow\pm\infty}e^{itH}\Pi_\mathrm{p}(H^\star)^\perp r_\mp(H)e^{-itH_0}.
\end{equation}
Here $\Pi_\mathrm{p}(H^\star)^\perp$ is the projection onto the orthogonal of the point spectral subspace of $H^\star$, and $r_\mp(H)$ are rational functions of $H$ regularizing the incoming/outgoing spectral singularities of $H$. See below for precise definitions. In our context, spectral singularities are defined as points of the essential spectrum of $H$ where $H$ does not satisfy a suitable limiting absorption principle (see Definition \ref{def:point_spectral_regulier_classique_pour_H} below). 
For Schr\"odinger operators with a bounded, compactly supported potential, spectral singularities correspond to real resonances (see e.g. \cite[chapter 3]{DyZw19_01} for a definition of resonances and \cite[section 5]{FaFr18_01} for a discussion on spectral singularities and real resonances).

We are not aware of such a definition of regularized wave operators previously in the literature, even in context of Schr\"odinger operators. In this paper, we prove that, under suitable assumptions, the regularized wave operators $W_\pm(H,H_0)$ exist, are injective, and such that their ranges are dense in $\Ran(\Pi_\mathrm{p}(H^\star)^\perp)$. (We will recall from \cite{FaFr22_06} that $\Ran(\Pi_\mathrm{p}(H^\star)^\perp)$ coincide with the `absolutely spectral subspace' of $H$, see Subsection \ref{subsec:orga} for the definition of $\mathcal{H}_{\mathrm{ac}}(H)$. These properties therefore extend the corresponding well-known properties that hold for wave operators in unitary scattering theory, which supports our definition of $W_\pm(H,H_0)$. 

Finally we show that if $H$ has no spectral singularity, then wave operators are \emph{asymptotically complete}, in the sense that they define bijections from $\Hi$ to $\Ran(\Pi_\mathrm{p}(H^\star)^\perp)$. A consequence of Asymptotic Completeness is that the solutions of the Schr\"odinger equation \eqref{eq:Sch equ} with initial conditions in $\Ran(\Pi_\mathrm{p}(H^\star)^\perp)$ are uniformly bounded in time. 

\medskip

\emph{Notation}.
In the following, if $\Hi_1$ and $\Hi_2$ are two Hilbert spaces, $\mathcal{B}(\Hi_1,\Hi_2)$ stands for the set of continuous linear operators from $\Hi_1$ to $\Hi_2$. If $\Hi_1=\Hi_2$, we simplify the notation by setting $\mathcal{B}(\Hi_1)=\mathcal{B}(\Hi_1,\Hi_1)$. We let  $\Res_B(z):=(B-z)^{-1}$ the resolvant of an operator $B$. For $H_0$, we denote $\Res_0(z)$ its resolvant. Moreover $D(\lambda,\varepsilon)$ is the open disk centered at $\lambda$ with radius $\varepsilon$ and $\C^\pm=\lbrace z\in\C, \pm\im(z)>0\rbrace$. Finally, the set of all integers between $1$ and $n$ is denoted $\llbracket 1,n\rrbracket$ and $A^{\mathrm{cl}}$ is the closure of the set $A$.

\section{Abstract setting}

\subsection{The model}

Consider $(\Hi,\scal{.}{.}_\Hi)$ a complex separable Hilbert space and an operator on $\Hi$ of the form 
\begin{equation}\label{eq:model def H}
	H:=H_0+V,
\end{equation}
where $H_0$ is self-adjoint and semi-bounded from below and $V\in\B(\Hi)$ is a bounded operator. In particular, $H$ is a closed operator with domain $\mathcal{D}(H)=\mathcal{D}(H_0)$ and its adjoint is given by 
\begin{equation*}
	H^\star=H_0+V^\star, \quad \mathcal{D}(H^\star)=\mathcal{D}(H_0).
\end{equation*}
Without loss of generality, we suppose that $H_0\geq 0$. 

As $H$ is a perturbation of the self-adjoint operator $H_0$ by the bounded operator $V$, $-iH$ generates a strongly continuous one-parameter group $\left\lbrace e^{-itH} \right\rbrace_{t\in\R}$ which satisfies
\begin{equation*}
	\norme{e^{-itH}}_{\B(\Hi)}\leq e^{\abso{t}\norme{V}_{\B(\Hi)}},\quad t\in\R,
\end{equation*}
(see e.g \cite{Da07_01} or \cite{EnNa20_01}).

We assume that there exists an operator $C\in\B(H)$, relatively compact with respect to $H_0$, such that $V$ is of the form
\begin{equation}\label{eq:model factorization V}
	V=CWC,
\end{equation}
with $W\in\B(\Hi)$. Finally, we require that $C$ be a \emph{metric operator}, which means that $C$ is a positive and injective (see e.g. \cite{AnTr12_12}).

\subsection{Spectral subspaces, spectral projection}

We let $\sigma(H)$ be the spectrum of $H$ and $\rho(H)=\C\setminus\sigma(H)$ its resolvent set. We define the point spectrum of $H$ as the set of all eigenvalues of $H$,
\begin{equation*}
	\sigma_{\mathrm{p}}(H):=\big\{\lambda \in \mathbb{C} , \, \Ker(H-\lambda)\neq\{0\}\big\}.
\end{equation*}
For each eigenvalue $\lambda$, we define its algebraic multiplicity $\mathrm{m}_\lambda(H)$ as the dimension of the generalized eigenspace associated to $\lambda$
\begin{equation*}
	\mathrm{m}_\lambda(H):=\mathrm{dim}\left(\bigcup_{k=1}^\infty\Ker\left( (H-\lambda)^k \right) \right). 
\end{equation*}
If $\lambda$ is an isolated eigenvalue of $H$, we denote by 
\begin{equation}\label{2eq:Projection_de_Riesz_def_pour_H}
	\Pi_\lambda(H):=\frac{1}{2\pi  i}\int_\gamma \left(z\Id-H\right)^{-1}\mathrm{d}z,
\end{equation}
the usual Riesz projection, where $\gamma$ is a circle oriented counterclockwise and centered at $\lambda$, of sufficiently small radius (so that $\lambda$ is the only point of the spectrum of $H$ contained in the interior of $\gamma$) is finite dimensional. The discrete spectrum of $H$, $\sigma_{\mathrm{disc}}(H)$, is the set of all isolated eigenvalues $\lambda$ such that the range of the associated Riesz projection is finite dimensional.

As $V$ is a relatively compact perturbation of $H_0$, the essential spectrum $\sigma_{\mathrm{ess}}(H):=\sigma(H)\backslash\sigma_\mathrm{disc}(H)$ and the essential spectrum of $H_0$ coincide. Moreover, the discrete spectrum $\sigma_{\mathrm{disc}}(H)$ is at most countable and can only accumulate at points of $\sigma_{\mathrm{ess}}(H)$. See Figure \ref{P2fig1}. We define in addition the set of eigenvalues embedded in the essential spectrum of $H$:
\begin{equation*}
	\sigma_{\mathrm{emb}}(H):=\sigma_{\mathrm{p}}(H)\cap\sigma_{\mathrm{ess}}(H).
\end{equation*}

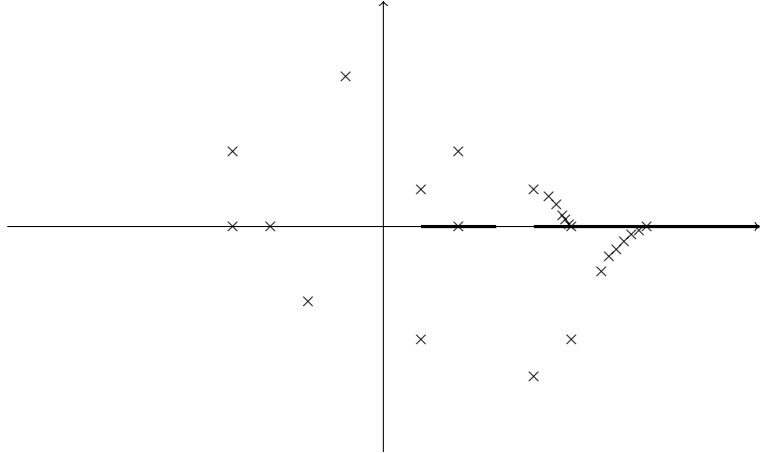
\begin{figure}
	\begin{center}
		\begin{tikzpicture}
			
			\draw[->] (-5,0) --(5,0);
			\draw[->] (0,-3) --(0,3);
			\draw[-, very thick] (0.5,0)--(1.5,0);
			\draw[-, very thick] (2,0)--(5,0);
			
			\draw (1,1) node {\tiny{$\times$}};
			
			\draw (0.5,0.5) node {\tiny{$\times$}};
			
			\draw (-2,1) node {\tiny{$\times$}};
			
			
			\draw (0.5,-1.5) node {\tiny{$\times$}};
			
			\draw (2,-2) node {\tiny{$\times$}};
			
			\draw (-1,-1) node {\tiny{$\times$}};
			
			\draw (-0.5,2) node {\tiny{$\times$}};
			
			\draw (2.5,-1.5) node {\tiny{$\times$}};

			\draw (1,0) node {\tiny{$\times$}};
			
			
			\draw(2,0.5) node {\tiny{$\times$}};
			
			\draw(2.2,0.4) node {\tiny{$\times$}};
			
			\draw(2.3,0.3) node {\tiny{$\times$}};
			
			\draw(2.38,0.15) node {\tiny{$\times$}};
			
			\draw(2.42,0.1) node {\tiny{$\times$}};
			
			\draw(2.47,0.03) node {\tiny{$\times$}};
			
			\draw (2.5,0) node {\tiny{$\times$}};
			
			\draw(-1.5,0)  node {\tiny{$\times$}};
			
			\draw(-2,0) node {\tiny{$\times$}};
			
			\draw(3.5,0) node {\tiny{$\times$}};
			
			\draw(3.4,-0.05)node {\tiny{$\times$}};
			
			\draw(3.3,-0.11)node {\tiny{$\times$}};
			
			\draw(3.2,-0.20) node {\tiny{$\times$}};
			
			\draw(3.1,-0.3) node {\tiny{$\times$}};
			
			\draw(3,-0.4) node {\tiny{$\times$}};
			
			\draw(2.9,-0.6) node {\tiny{$\times$}};

		\end{tikzpicture}
	\end{center}
	\caption{ \footnotesize  \textbf{Spectrum of $H$.} The essential spectrum of $H$, contained in $[ 0 , \infty )$, is represented by thick lines and coincides with that of $H_0$. The eigenvalues of $H$ are represented by crosses. The discrete spectrum of $H$ consists of isolated eigenvalues of finite algebraic multiplicities which may accumulate at any point of the essential spectrum. The point spectrum of $H$ may also contain eigenvalues embedded in the essential spectrum. }\label{P2fig1}
\end{figure}

\subsubsection{Eigenspaces corresponding to isolated eigenvalues}\label{P2subsubsec:isolated}
For $\lambda\in\sigma_\mathrm{disc}(H)$, since the restriction of $H$ to $\Ran(\Pi_\lambda(H))$ may have a nontrivial Jordan form, $\Ran(\Pi_\lambda(H))$ is in general spanned by generalized eigenvectors of $H$ associated to $\lambda$, i.e., by vectors $u\in\mathcal{D}(H^k)$ such that $(H-\lambda)^{k}u=0$ for some $1\leq k\leq \mathrm{m}_\lambda(H)$. Moreover the dimension of the range of $\Pi_\lambda(H)$ satisfies $\dim\Ran(\Pi_\lambda(H))=\mathrm{m}_\lambda(H)$.
We set 
\begin{equation*}
	\Hi_\mathrm{disc}(H):=\Span\left\lbrace u\in\Ran(\Pi_\lambda(H)),~\lambda\in\sigma_\mathrm{disc}(H)\right\rbrace^{\mathrm{cl}},
\end{equation*}
where $A^{\mathrm{cl}}$ stands for the closure of a subset $A\subset\Hi$. We will sometimes assume that the discrete spectrum of $H$ is finite. The spectral projection $\Pi_{\mathrm{disc}}(H)$ onto $\Hi_\mathrm{disc}(H)$ is then defined by
\begin{equation}\label{P2eq:def_Pi_disc}
	\Pi_{\mathrm{disc}}(H):=\sum_{\lambda\in\sigma_{\mathrm{disc}}(H)}\Pi_\lambda(H).
\end{equation}

  \subsubsection{Eigenspaces corresponding to embedded eigenvalues}\label{subsubsec:embedded}
		If $\lambda$ is an eigenvalue of $H$ embedded in its essential spectrum then the Riesz projection corresponding to $\lambda$ is ill-defined. Under some additional conditions, however, one can define the spectral projection $\Pi_\lambda(H)$ as follows.
		
		In the following (see Hypothesis \ref{hyp:Conjugate operator} below), we will suppose the existence of a \emph{antiunitary operator} $J\in\mathcal{B}(\Hi)$ verifying
			\begin{equation}\label{eq:existence_J}
			 J\mathcal{D}(H_0)\subset\mathcal{D}(H_0) \quad\text{and}\quad	\forall u\in\mathcal{D}(H_0), \quad JHu=H^\star Ju.
			\end{equation}
		In particular, $J$ establishes a one-to-one correspondence between $\Ker((H-\lambda)^{k})$ and $\Ker((H^\star-\bar{\lambda})^{k})$ for all $k\in\mathbb{N}$ and hence
		\begin{equation*}
		\mathrm{m}_\lambda(H)=\mathrm{m}_{\bar{\lambda}}(H^\star).
		\end{equation*} 
		To shorten notation, let $\mathrm{m}_\lambda=\mathrm{m}_\lambda(H)=\mathrm{m}_{\bar\lambda}(H^\star)$. In order to define projections onto the generalized eigenspace associated to embedded eigenvalues, we will suppose that for each embedded eigenvalue $\lambda\in\sigma_{\mathrm{ess}}(H)$, $\mathrm{m}_\lambda$ is finite and the symmetric bilinear form
		\begin{equation}\label{eq:invertibility_mat}
			\Ker((H-\lambda)^{\mathrm{m}_\lambda})\times	\Ker((H-\lambda)^{\mathrm{m}_\lambda}) \ni (u,v)\mapsto \langle Ju,v\rangle_\Hi \quad\text{is non-degenerate}.
		\end{equation}
		This implies that there exists a basis $(\varphi_k)_{1\leq k\leq \mathrm{m}_\lambda}$ of $\Ker((H-\lambda)^{\mathrm{m}_\lambda})$ such that
		\begin{equation*}
		\langle J \varphi_i , \varphi_j \rangle_{\Hi} = \delta_{ij} , \quad 1\le i,j\le m_\lambda.
		\end{equation*}
		Thus we can define the spectral projection $\Pi_\lambda(H)$ onto the generalized eigenspace corresponding to $\lambda$ as
		\begin{equation}\label{eq:proj_emb}
			 \Pi_\lambda(H)u = \sum_{k=1}^{\mathrm{m}_\lambda}\scal{J\varphi_k}{u}_{\Hi}\varphi_k, \quad u \in \Hi.
		\end{equation}
		It is not difficult to observe that $\Pi_\lambda(H)$ is a projection commuting with $H$, such that $\Pi_\lambda(H)\in\mathcal{B}(\Hi)$ and $\Pi_\lambda(H)^\star=\Pi_\lambda(H^\star)$. For more detail about this projection, see \cite{FaFr22_06}. 
		
		\subsubsection{Projection onto the point spectrum}
		
		In the following we will suppose that $H$ has a finite number of eigenvalues with finite algebraic multiplicities, see Hypothesis \ref{hyp:VPH}. Thus the sum of all the projections associated to generalized eigenspaces of $H$, defined by 
		\begin{equation*}
			\Pi_\mathrm{p}(H):=\sum_{\lambda\in\sigma_\mathrm{p}(H)}\Pi_\mathrm{\lambda}(H)		
		\end{equation*}
		is well-defined. Next we define the point spectral subspace of $H$ as the range of $\Pi_\mathrm{p}(H)$, 
		\begin{equation*}
			\Hi_\mathrm{p}(H):=\Ran(\Pi_\mathrm{p}(H)). 
		\end{equation*}
			Finally, we observe that 
		\begin{equation*}
			\Hi_\mathrm{p}(H)=\displaystyle\sum_{\lambda\in\sigma_\mathrm{p}(H)}\Ker(H-\lambda)^{\mathrm{m}_{\lambda}}.
		\end{equation*}

\section{Assumptions and main results}

\subsection{Hypotheses}	

In this section we detail our main abstract assumptions. In Section \ref{subsec:Appli} we will show that they are satisfied in the case of complex Schr\"odinger operators, with compactly supported potentials.

In our first hypothesis, we require that, at any point of its essential spectrum, $H_0$ satisfies a \emph{limiting absorption principle} with weight $C$.

\begin{assumption}[Limiting absorption principle for $H_0$]\label{hyp:LAP H0}
		We have
	\begin{equation}\label{eq:LAP_H0}
	\sup_{z\in\C^\pm}\big\|C\Res_0(z)C\big\|_{\mathcal{B}(\Hi)}<\infty.
	\end{equation}
\end{assumption} 

	Note that \eqref{eq:LAP_H0} implies (see e.g. \cite[Proposition 4.1]{CFKS}) that the spectrum of $H_0$ is purely absolutely continuous, i.e. that $\sigma_{\mathrm{pp}}(H_0)=\emptyset$, $\sigma_{\mathrm{ac}}(H_0)=\sigma(H_0)$, $\sigma_{\mathrm{sc}}(H_0)=\emptyset$, where $\sigma_{\mathrm{pp}}(H_0)$, $\sigma_{\mathrm{ac}}(H_0)$, $\sigma_{\mathrm{sc}}(H_0)$ stand for the usual pure point, absolutely continuous and singular continuous spectra of the self-adjoint operator $H_0$. 

	Moreover, by Fatou's Theorem (see \cite{Rud1980}), \eqref{eq:LAP_H0} implies that the limits $C\Res_0(\lambda\pm i0^+)C$ exist for almost every $\lambda\in\sigma_{\mathrm{ess}}(H)$, in the norm topology of $\mathcal{B}(\Hi)$, and that the map $\R\ni\lambda\mapsto C\Res_0(\lambda\pm i0^+)C\in\mathcal{B}(\Hi)$ is bounded (observe that $C\Res_0(\lambda\pm i0^+)C=C\Res_0(\lambda)C$ if $\lambda\in\R\setminus\sigma_{\mathrm{ess}}(H)$).
	
Note also that Hypothesis \ref{hyp:LAP H0} implies (see \cite{Ka65_01} or \cite[Theorem XIII.25 and its corollary]{ReSi80_01}) that $C$ is \emph{relatively smooth} with respect to $H_0$ in the sense of Kato, i.e. that there exists a constant $c_0$ such that 
		\begin{equation}\label{eq:Kato_smooth}
			\forall u\in\Hi, \quad \int_\R\norme{Ce^{-itH_0}u}_\Hi^2\mathrm{d}t\leq c_0^2\norme{u}_\Hi^2.
		\end{equation}
		Recall that \eqref{eq:Kato_smooth} is equivalent to 
		\begin{equation}\label{eq:Smooth_Ka65_01}
			\forall u\in\Hi, \quad \int_\R\left(\norme{C\Res_0(\lambda-i0^+)u}_\Hi^2+\norme{C\Res_0(\lambda+i0^+)u}_\Hi^2\right)\mathrm{d}\lambda\leq 2\pi c_0^2\norme{u}_\Hi^2,
		\end{equation}
		where $\lambda\mapsto C\Res_0(\lambda\pm i0^+)u$ denotes the limit of $\lambda\mapsto C\Res_0(\lambda\pm i\varepsilon)u$ in $L^2(\R;\Hi)$ as $\varepsilon\to0^+$.

Next we assume that the point spectral subspace of $H$ is finite.
		
	\begin{assumption}[Eigenvalues of $H$]\label{hyp:VPH}
		$H$ has only a finite number of eigenvalues with finite algebraic multiplicities. 
	\end{assumption}

Hypothesis \ref{hyp:VPH} prevents the essential spectrum of $H$ from having an accumulation point of eigenvalues. It does not exclude, however, the presence of eigenvalues embedded in the essential spectrum of $H$.

Our next hypothesis concerns the spectral singularities of $H$. First we recall the definition of a regular spectral point. Note that this one is independent of the previous assumption. 

\begin{definition}[Regular spectral point and spectral singularity]\label{def:point_spectral_regulier_classique_pour_H}$ $
	\begin{enumerate}
	\item 	Let $\lambda\in\sigma_\mathrm{ess}(H)$. 
	\begin{enumerate}[label=(\roman*)]
		\item We say that $\lambda$ is an outgoing/incoming regular spectral point of $H$ if $\lambda$ is not an accumulation point of eigenvalues located in $\lambda\pm i\left( 0,\infty\right)$ and if the limit 
		\begin{equation}\label{eq:def_reg_spec_pt}
			C\Res_H(\lambda\pm i0^+)CW:=\lim_{\varepsilon\rightarrow 0^+} C\Res_H(\lambda\pm i\varepsilon)CW
		\end{equation}
		exists in the norm topology of $\mathcal{B}(\Hi)$. If $\lambda$ is not an outgoing/incoming regular spectral point, we say that $\lambda$ is an outgoing/incoming spectral singularity of $H$.
		\item We say that $\lambda$ is a regular spectral point of $H$ if $\lambda$ is both an outgoing and an incoming regular spectral point of $H$. If $\lambda$ is not a regular spectral point, we say that $\lambda$ is a \emph{spectral singularity} of $H$.
	\end{enumerate}
	\item We say that infinity is an outgoing/incoming regular spectral point of $H$ if there exists $m>0$ such that for all $\lambda>m$, $\lambda$ is an outgoing/incoming regular spectral point 
	and such that the map 
	\begin{equation*}
		[m,\infty)\ni\lambda\mapsto C\Res_H(\lambda\pm i0^+)CW\in\mathcal{B}(\Hi)
	\end{equation*}
	is bounded. 
	If infinity is not outgoing/incoming regular spectral point of $H$, we say that $H$ has an outgoing/incoming spectral singularity at infinity. 
	\end{enumerate}
	\end{definition}	

Next, we introduce the notion of \emph{order} of a spectral singularity. 

\begin{definition}[Order of a spectral singularity]$ $
	\begin{enumerate}
	\item Let $\lambda\in\sigma_\mathrm{ess}(H)$ be an outgoing/incoming spectral singularity of $H$. We say that $\lambda$ is a spectral singularity of finite order if there exist an integer $n$ and $\varepsilon>0$ such that 
	\begin{equation}\label{eq:def order of spectral singularity}
		\sup_{z\in D(\lambda,\varepsilon)\cap\C^\pm}\abso{\lambda-z}^n\norme{C\Res_H(z)CW}_{\mathcal{B}(\Hi)}<\infty.
	\end{equation}
	Otherwise we say that $\lambda$ is an outgoing/incoming spectral singularity of infinite order. If $\lambda$ is an outgoing/incoming spectral singularity of finite order we define its order as the smallest integer satisfying \eqref{eq:def order of spectral singularity}. 
	\item If $H$ has an outgoing/incoming spectral singularity at infinity, we say that infinity is an outgoing/incoming spectral singularity of finite order if there exists an integer $n$, $\varepsilon_0>0,m>0$ and $z_0\in\rho(H)\backslash\R$, such that 
	\begin{equation}\label{eq:def order spectral singularity infinity}
		\sup_{\substack{\re(z)>m \\ \pm\im(z)<\varepsilon_0}}\abso{z-z_0}^{-n}\norme{C\Res_H(z)CW}_{\mathcal{B}(\Hi)}<\infty.
	\end{equation}
	Otherwise we say that infinity is an outgoing/incoming spectral singularity of infinite order. If infinity is an outgoing/incoming spectral singularity of finite order we define its order as the smallest integer satisfying \eqref{eq:def order spectral singularity infinity}. 
\end{enumerate}
\end{definition}

Remark that $\lambda$ is a spectral singularity of finite order means that, in a neighborhood of $\lambda$, the map $z\mapsto C\Res_H(z)CW$ blows up at most polynomially as $z$ approaches $\lambda$. Moreover if the map $z\mapsto C\Res_H(z)CW$ has a meromorphic continuation across $\sigma_\mathrm{ess}(H)$, then the spectral singularities of $H$ correspond to poles of the meromorphic continuation of the weighted resolvent of $H$, and the order of the spectral singularity correspond to the order of the pole. This is in particular the case for Schr\"odinger operators (see \ref{subsec:Appli}).

In the following, we assume that $H$ has a finite number of spectral singularities and that each spectral singularity has a finite order.

	\begin{assumption}[Spectral singularities for $H$]\label{hyp:spectral singularities}
		$H$ only has a finite number of outgoing/incoming spectral singularities in $\sigma_{\mathrm{ess}}(H)\cup\{\infty\}$ and each spectral singularity has a finite order.  Moreover, for all closed  interval $I\subset\sigma_{\mathrm{ess}}(H)$ not containing any spectral singularity, there exists $\varepsilon_0>0$ such that 
		\begin{equation*}
			\sup_{\substack{\mathrm{Re}(z)\in I\\ \pm\mathrm{Im}(z)\in(0,\varepsilon_0)}}  \big\|C\Res_H(z)CW\big\|_{\mathcal{B}(\Hi)}<\infty.
		\end{equation*}
	\end{assumption}
	
This hypothesis has the following consequence. Let $\lambda_1,\ldots\lambda_n\in \sigma_{\mathrm{ess}}(H)$  be the spectral singularities of $H$ belonging to $\sigma_{\mathrm{ess}}(H)$, of order $\nu_1,\dots,\nu_n<\infty$, respectively, and let $\nu_\infty$ be the order of $\infty$ in the case where $\infty$ is a spectral singularity (otherwise, we use the convention that $\nu_\infty=0$). Then there exists $\varepsilon_0>0$ such that 
		\begin{equation}\label{eq:limit_uniform}
			\sup_{\substack{\mathrm{Re}(z)\in \sigma_{\mathrm{ess}}(H)\\ \pm\mathrm{Im}(z)\in(0,\varepsilon_0)}} \frac{1}{|z-z_0|^{\nu_\infty}} \Big( \prod_{j=1}^n\frac{ |z-\lambda_j|^{\nu_j} }{ |z-z_0|^{\nu_j} } \Big ) \big\|C\Res_H(z)CW\big\|_{\mathcal{B}(\Hi)}<\infty,
		\end{equation}
	where $z_0$ is an arbitrary complex number such that $z_0\in\rho(H)$, $z_0\in\C\setminus\R$. 
	Note that the factors $|z-\lambda_j|^{\nu_j}$ `regularize' the singularities of $z\mapsto C\Res_H(z)CW$ as $z$ approaches $\lambda_j$. Dividing them by $|z-z_0|^{\nu_j}$ produces bounded terms. The factor $|z-z_0|^{-\nu_\infty}$ regularizes a possible singularity at $\infty$.
	
	Observe that since $\lambda_1,\dots,\lambda_n$ are the only spectral singularities of $H$, for all $\lambda\in \sigma_{\mathrm{ess}}(H)\setminus\{\lambda_1,\dots,\lambda_n\}$, the limits $C\Res_H(\lambda\pm i0^+)CW$ exist in the norm topology of $\mathcal{B}(\Hi)$. The condition \eqref{eq:limit_uniform} then implies that the maps
	\begin{equation}
		\sigma_{\mathrm{ess}}(H)\setminus\{\lambda_1,\dots,\lambda_n\} \ni \lambda \mapsto \frac{1}{|\lambda-z_0|^{\nu_\infty}} \Big( \prod_{j=1}^n\frac{ |\lambda-\lambda_j|^{\nu_j} }{ |\lambda-z_0|^{\nu_j} } \Big ) C\Res_H(\lambda\pm i0^+)CW \in \mathcal{B}(\Hi)
	\end{equation}
	are bounded. Since embedded eigenvalues are outgoing and incoming spectral singularities (see \cite{FaFr22_06}), Hypothesis \ref{hyp:spectral singularities} also concern possible embedded eigenvalues. In the following we will denote by $r_j$ the function that regularizes the spectral singularity $\lambda_j$, i.e.
	\begin{equation}\label{eq:rj}
		r_j(z):=\frac{(z-\lambda_j)^{\nu_j}}{(z-z_0)^{\nu_j}}.
	\end{equation} 
In the same way, $r_\infty$ stands for the function that regularizes the spectral singularity at $\infty$, i.e.
\begin{equation}\label{eq:rinfty}
	r_\infty(z):=(z-z_0)^{-\nu_\infty}.
\end{equation}
 Moreover we let $r_+$ be the product of all the functions $r_j$ where $\lambda_j$ is an outgoing spectral singularity, and $r_-$ be the product of all the functions $r_j$ where $\lambda_j$ is an incoming spectral singularity. 

We also mention the following `local version' of \eqref{eq:limit_uniform} which will be useful in the sequel. Obviously, since $H$ has a finite number $n$ of spectral singularities, it also has a finite number of outgoing and incoming spectral singularities. Let $J_1,\dots,J_n$ be compact intervals and $J_\infty$ be an interval such that i) for each $j\in\llbracket 1,n\rrbracket$, the interior $J_j^\mathrm{int}$ of $J_j$ contains the outgoing/incoming spectral singularity $\lambda_j$ and no other outgoing/incoming spectral singularity, ii) $J_1^\mathrm{int},\dots,J_n^\mathrm{int},J_\infty^\mathrm{int}$ are disjoint, and iii) $\sigma_\mathrm{ess}(H)=\bigcup_{j=1}^n J_j\cup J_\infty$. As, for $j\in\llbracket 1,n\rrbracket\cup\lbrace \infty\rbrace$, $r_k$ is invertible in $J_j$ for all $k\in \llbracket 1,n\rrbracket\cup\lbrace\infty\rbrace \backslash \lbrace j\rbrace$, we deduce from \eqref{eq:limit_uniform} the following local limiting absorption principle:
\begin{equation*}
\sup_{\substack{\re(z)\in J_j\\ \pm \im(z)\in (0,\varepsilon_0)}}\abso{r_j(z)}\norme{C\Res_H(z)CW}_{\mathcal{B}(\Hi)}<\infty.
\end{equation*}

Finally, in order to define the spectral projections onto the generalized eigenspaces associated to embedded eigenvalues, we require the existence of a conjugation operator $J$ satisfying, in particular, $JH=H^\star J$. 
	\begin{assumption}[Conjugation operator and embedded eigenvalues]\label{hyp:Conjugate operator}
		There exists an antiunitary operator $J:\Hi\to\Hi$ such that 
		\begin{enumerate}[label=(\roman*)]
			\item $J\mathcal{D}(H_0)\subset\mathcal{D}(H_0)$ and $\forall u\in\mathcal{D}(H_0)$, $JH_0u=H_0Ju$.
			\item $JC=CJ$ and $JW=W^\star J$. 
		\end{enumerate}
		Moreover, for all embedded eigenvalues $\lambda\in\sigma_{\mathrm{ess}}(H)$, the symmetric bilinear form
		\begin{equation}\label{eq:invertibility_mat2}
			\Ker\big((H-\lambda)^{\mathrm{m}_\lambda}\big)\times\Ker\big((H-\lambda)^{\mathrm{m}_\lambda}\big) \ni (u,v)\mapsto \langle Ju,v\rangle \quad\text{is non-degenerate}.
		\end{equation}
	\end{assumption}
	This condition ensures that the spectral projections \eqref{eq:proj_emb} onto the generalized eigenspaces associated to embedded eigenvalues are well-defined. See the discussion in \cite[section 2.2.2]{FaFr22_06} for more details.

\subsection{Main results}

In this section we state our main results. Recall from \eqref{eq:def_WO-intro} that the wave operators associated to $H$ and $H_0$ are defined by 
\begin{equation*}
W_\pm(H,H_0):=\slim_{t\rightarrow\pm\infty} e^{itH}\Pi_\mathrm{p}(H^\star)^\perp r_\mp(H)e^{-itH_0}
\end{equation*} 
We recall that $\Pi_\mathrm{p}(H^\star)$ is the sum of all the spectral projections onto the generalized eigenspaces associated to the eigenvalues of $H^\star$ (either isolated or embedded into the essential spectrum), and that $\Pi_\mathrm{p}(H^\star)$ is well-defined thanks to Hypotheses \ref{hyp:VPH} and \ref{hyp:Conjugate operator}. The projection  $\Pi_\mathrm{p}(H^\star)^\perp = \mathrm{Id}-\Pi_\mathrm{p}(H)$ is needed to project out the generalized eigenspace associated to $H$, and the `regularizing operator' $r_\mp(H)$ is needed to regularize the incoming/outgoing spectral singularities of $H$. 
See section \ref{subsubsec:Justification of the definition} for more details about the justification of this definition. 
In the same way as for \eqref{eq:def_WO-intro}, this `non-stationary' definition of wave operators taking into account possible spectral singularities seem to be new. See however e.g. \cite{Ka65_01} for general definitions of non necessarily unitary wave operators, under assumptions ensuring the absence of spectral singularities.    

Our first main result is the existence of the wave operators $W_\pm(H,H_0)$ as well as properties about their ranges and kernels. 
\begin{theorem}\label{thm:existence of the wave operator}
Suppose that Hypotheses \ref{hyp:LAP H0}-\ref{hyp:Conjugate operator} hold. Then the wave operators $W_\pm(H,H_0)$ exist. Moreover they are injective and their ranges are dense in $\Hi_\mathrm{p}(H^\star)^\perp$. 
\end{theorem}

Once the existence of the wave operators $W_\pm(H,H_0)$ is proven, it is not difficult to show that they satisfy the usual intertwining properties (see Proposition \ref{prop:interwining property}). 

One of the ingredients to the properties of these operators is the studying of
\begin{equation*}
	W_\pm(H_0,H):=\slim_{t\rightarrow\pm\infty}e^{itH_0}r_\pm(H)\Pi_\mathrm{p}(H)^\perp e^{-itH}
\end{equation*}
in section \ref{subsec:autre operateur onde}. Because $\lbrace e^{-itH_0}\rbrace_{t\in\R}$ is unitary, the existence of $W_\pm(H_0,H)$ implies that for all $u\in\Ran{W_\pm(H_0,H)}$ there exists $v\in\Hi$ such that 
\begin{equation}\label{eq:scatt H-H_0}
	\lim_{t\rightarrow\pm\infty}\norme{e^{-itH}\Pi_\mathrm{p}(H^\star)^\perp r_\pm(H)v-e^{-itH_0}u}=0.
\end{equation}
And \eqref{eq:scatt H-H_0} leads to $u=W_\pm(H_0,H)v$ and
\begin{equation*}
	e^{-itH_0}u=\Pi_\mathrm{p}(H^\star)^\perp r_\pm(H)e^{-itH}v+o(1),\quad t\rightarrow \pm\infty.
\end{equation*}
Next we want to study the inversibility of the wave operator $W_\pm(H,H_0)$. We say that $W_\pm(H,H_0)$ is asymptotically complete if its range is closed. In particular, if the wave operator $W_\pm(H,H_0)$ is asymptotically complete, then it is invertible in $\mathcal{B}(\Hi,\Hi_\mathrm{p}(H^\star)^\perp)$. The next theorem give sufficient conditions for asymptotic completeness. 

\begin{theorem}\label{thm3.4}
	Suppose that Hypotheses \ref{hyp:LAP H0}-\ref{hyp:Conjugate operator} hold and that $H$ does not have any spectral singularity. Then the wave operators $W_\pm(H,H_0)$ are asymptotically complete. 
\end{theorem}

A similar result is proved in \cite{FaFr18_01} for dissipative operators. In particular, the previous theorem implies that, if $H$ does not have any spectral singularity, then $H$ and $H_0$ are similar in the sense that 
\begin{equation*}
	\forall u\in\Hi_\mathrm{p}(H^\star)^\perp\cap\mathcal{D}(H_0),\quad H_0u=W_\pm(H,H_0)^{-1}HW_\pm(H,H_0)u.
\end{equation*}

Another consequence of Theorem \ref{thm3.4} is the following result.

\begin{theorem}\label{mainthm:boundedness of solution of Schrodinger equation}
	Suppose that Hypotheses \ref{hyp:LAP H0}-\ref{hyp:Conjugate operator} hold. Suppose that $H$ does not have any spectral singularities. Then there exist $m_1>0$, and $m_2>0$ such that 
	\begin{equation*}
		\forall t \in \R,\quad \forall u\in\Hi_\mathrm{p}(H^\star)^\perp, \quad m_1\norme{u}_\Hi\leq \norme{e^{-itH}u}_\Hi\leq m_2\norme{u}_\Hi. 
	\end{equation*}
\end{theorem}
The second inequality shows that solutions of the Schr\"odinger equation
\begin{equation*}
	\begin{cases}
		&i\partial_tu_t=Hu_t\\
		&u_0\in\Hi_\mathrm{p}(H^\star)^\perp
	\end{cases}
\end{equation*}
cannot blow up as $t\rightarrow \pm\infty$ and that the norm of $u_t$ is controlled by the norm of the initial state $u_0$. Related results have been established in \cite{Go10_01} for non-self-adjoint Schr\"odinger operators and in \cite{FaFr18_01} for abstract dissipative operators.

Finally, using the definition of spectral projection associated to $H$ in an interval $I\subset\sigma_\mathrm{ess}(H)$ without spectral singularities (see \eqref{eq:def spectral projection} below), we can define the local wave operators associated to $H$ and $H_0$ on $I$ by
\begin{equation*}
	W_\pm(H,H_0,I):=\slim_{t\rightarrow\pm \infty}e^{itH}\mathds{1}_I(H)e^{-itH_0}.
\end{equation*}
The next theorem then shows that $W_\pm(H,H_0,I)$ are asymptotically complete in the sense that the range of $W_\pm(H,H_0,I)$ is equal to the range of $\mathds{1}_I(H)$. 

\begin{theorem}\label{thm:local asymptotic completeness }
	Suppose that Hypotheses \ref{hyp:LAP H0}-\ref{hyp:spectral singularities} hold. Let $I\subset\sigma_\mathrm{ess}(H)$ be a closed interval without any spectral singularity. Then $W_\pm(H,H_0,I)$ is invertible in $\mathcal{B}(\Ran(\mathds{1}_I(H_0)),\Ran(\mathds{1}_I(H)))$. 
\end{theorem}

Similar results have been obtained in \cite{FaFr18_01} in abstract dissipative scattering theory, and in \cite{KaYa76_09} for a class of non-self-adjoint Schr\"odinger operators. 

\subsection{Applications to Schr\"odinger operators}\label{subsec:Appli}

In this section we will show that our results apply to non-self-adjoint Schr\"odinger operators. We consider 
\begin{equation*}
	\Hi=L^2(\R^\mathrm{d},\C),\quad d\geq 3 \text{ odd}, \quad H_0=-\Delta, \quad \mathcal{D}(H_0)=H^2(\R^d,\C),
\end{equation*} 
where $H^2(\R^d,\C)$ is the usual Sobolev's space. It is well-known that $\sigma_\mathrm{ess}(H)=[0,\infty)$ and that
\begin{equation*}
	\sup_{z\in\C\backslash\R}\norme{\langle x\rangle^{-\delta}(-\Delta-z)^{-1}\langle x\rangle^{-\delta}}_{\mathcal{B}(\Hi)}<\infty,\quad \delta>1,
\end{equation*}
where $\langle x\rangle=(1+x^2)^{1/2}$. Hence choosing $C(x)=\langle x\rangle^{-\delta}$ with $\delta>1$, we see that Hypothesis \ref{hyp:LAP H0} is satisfied.

We suppose that 
\begin{equation}\label{eq: V compactly supported}
	V\in L^\infty_{\mathrm{c}}(\R^d,\C)=\lbrace u\in L^\infty(\R^d)\text{ compactly supported} \rbrace.
\end{equation} 
It then follows from \cite{FrLaSa16_01} that $H$ has only a finite number of discrete eigenvalues with finite algebraic multiplicities. In particular, Hypothesis \ref{hyp:VPH} is satisfied. 

To verify that Hypothesis \ref{hyp:spectral singularities} holds, we rely on the theory of resonances. Resonances are defined as poles of the meromorphic extension of the weighted resolvent, see e.g \cite{DyZw19_01}. More precisely, assuming \eqref{eq: V compactly supported}, the map 
\begin{equation*}
	\lbrace z\in \C,\im(z)>0\rbrace \ni z\mapsto(H-z^2)^{-1}:L^2(\R^d)\rightarrow L^2(\R^d)
\end{equation*}
is meromorphic and extends to a meromorphic map 
\begin{equation*}
	\C\ni z\mapsto(H-z^2)^{-1}:L_\mathrm{c}^2(\R^d)\rightarrow L_\mathrm{loc}^2(\R^d),
\end{equation*}
where where $L^2_{\mathrm{c}}(\mathbb{R}^d):=\{u\in L^2(\mathbb{R}^d),u\text{ is compactly supported} \}$ and $L^2_{\mathrm{loc}}( \mathbb{R}^d ) :=\{u:\mathbb{R}^d\to\mathbb{C},u\in L^2(K) \text{ for all compact set } K \subset \mathbb{R}^d \}$. Poles of the extension of the resolvent of $H$ are called resonances. To a real resonance $\pm \lambda_0$ with $\lambda_0\geq 0$ corresponds an outgoing or/and an incoming spectral singularity $\lambda_0^2$ in the sense of Definition \ref{def:point_spectral_regulier_classique_pour_H}. The order of the pole of the resonance coincides with the order of the spectral singularity.  Hence, if for example $\lambda_0$ is an outgoing spectral singularity, there exists $\mathcal{V}_{\lambda_0}$ a closed real neighborhood of $\lambda_0$ whose interior contains $\lambda_0$ and no other resonance, such that 
\begin{equation*}
	\sup_{\substack{\re(z)\in \mathcal{V}_\lambda\\ \im(z)<\varepsilon_0}}\abso{\lambda_0-z}^{\nu_0}\norme{C\Res_H(z)CW}_{\mathcal{B}(\Hi)}<\infty
\end{equation*}
where $\nu_0$ is the order of the pole of $\lambda_0^2$ and $\varepsilon_0>0$ is small enough. Furthermore, under the assumption \eqref{eq: V compactly supported}, $H$ only has finitely many spectral singularities $\lbrace \lambda_1,\ldots,\lambda_n\rbrace$ of finite orders, and no spectral singularity at infinity. This shows that Hypothesis \ref{hyp:spectral singularities} is satisfied. We refer to \cite{DyZw19_01} for more details about the theory of resonances and \cite{FaFr18_01} for more details about the relation between resonances and spectral singularities.

Applying Theorems \ref{thm:existence of the wave operator}, \ref{thm3.4} and \ref{mainthm:boundedness of solution of Schrodinger equation}, we obtain the following proposition. 

\begin{proposition}
	Suppose that $V\in L^\infty_\mathrm{c}(\R^d,\C)$ and that for each eigenvalue embedded in the essential spectrum of $H$, the symmetric bilinear form 
	\begin{equation}\label{eq:non deg Schr}
		\Ker\big((H-\lambda)^{\mathrm{m}_\lambda}\big) \ni (u,v)\mapsto \int_{\mathbb{R}^d}u(x)v(x)\mathrm{d}x
	\end{equation}
	is non-degenerate. Then the wave operators $W_\pm(H,H_0)$ associated to $H=-\Delta+V$ and $H_0=-\Delta$ exist, are injective and their ranges are dense in $\Hi_\mathrm{p}(H)^\perp$.
	
	Moreover if $H$ does not have any real resonance, then $W_\pm(H,H_0)$ are asymptotically complete and there exist $m_1>0$ and $m_2>0$ such that, for all $u\in\Hi_\mathrm{p}(H)^\perp$, 
	\begin{equation*}
		\forall t\in\R,\quad m_1\norme{u}_\Hi\leq \norme{e^{itH}u}_{\Hi}\leq m_2\norme{u}_\Hi.
	\end{equation*}
\end{proposition}

Note that if $H$ has no embedded eigenvalue, then the condition \eqref{eq:non deg Schr} is not necessary. We mention that the condition \eqref{eq:non deg Schr} has been considered in \cite{Wa20_01} and \cite{Aa21_01}, to study the long-time behavior of solutions to the Schr\"odinger equation with a complex potential.


Finally we apply Theorem \ref{thm:local asymptotic completeness }. 

\begin{proposition}
	Suppose that $V\in L^\infty_\mathrm{c}(\R^d,\C)$. Let $I\subset [0,\infty)$ be a closed interval without any resonance. Then the local wave operators $W_\pm(H,H_0,I)$ exist and are asymptotically complete.
\end{proposition}

\subsection{Organization of the paper and main ingredients}\label{subsec:orga}

In this section we describe some of the main tools we will use in this paper.

To prove the existence of the wave operators $W_\pm(H,H_0)$, it is useful to characterize the spectral subspace $\Hi_\mathrm{p}(H^\star)^\perp$. We define the absolutely continuous spectral subspace of $H$, denoted $\Hi_\mathrm{ac}(H)$, as the closure of $\mathcal{M}(H)$ in $\Hi$, where
\begin{equation}\label{eq:def_M(H)}
	\mathcal{M}(H):=\left\lbrace u\in\Hi,\forall v\in\Hi,\int_\R\abso{\scal{e^{itH}u}{v}_\Hi}^2\mathrm{d}t\leq c_u\norme{v}_\Hi^2\right\rbrace. 
\end{equation}
It follows from \cite{FaFr22_06} that, under our assumptions, $\Hi_\mathrm{ac}(H)$ coincide with $\Hi_\mathrm{p}(H^\star)^\perp$. We will see that there exists $c>0$ such that, for all $u\in\Hi_\mathrm{ac}(H)$,
\begin{equation*}
	\int_0^\infty\norme{Cr_\pm(H)e^{\mp itH}u}_\Hi^2\mathrm{d}t\leq c\norme{u}_\Hi^2.
\end{equation*}
The existence of the wave operators will then follows from this property, the equality $\Hi_\mathrm{ac}(H)=\Hi_\mathrm{p}(H^\star)^\perp$ and Cook's method (see \cite{ReSi80_01} or \cite{YA98}).

The properties concerning the kernel and the range of the regularized wave operators mainly follow from the direct sum decomposition, 
\begin{equation*}
	\Hi_\mathrm{ac}(H)\oplus_J\Hi_\mathrm{p}(H)=\Hi,
\end{equation*}
proven in \cite{FaFr22_06}. This is a generalization of the well-known spectral decomposition for self-adjoint operators without singular continuous spectrum. Here $\oplus_J$ stand for the orthogonal direct sum for the bilinear form $\scal{J.}{.}_\Hi$. 

If $I\subset\sigma_\mathrm{ess}(H)$ is a closed interval without any spectral singularity, we can define the spectral projection associated to $H$ in $I$ by setting
\begin{equation}\label{eq:def spectral projection}
	\mathds{1}_I(H):=\frac{1}{2\pi i}\wlim_{\varepsilon\rightarrow 0^+}\int_I\left(\Res_H(\lambda+i\varepsilon)-\Res_H(\lambda-i\varepsilon)\right)\mathrm{d}\lambda.
\end{equation}
It follows from \cite{FaFr22_06} that, under our assumptions, \eqref{eq:def spectral projection} exists and is a projection. Moreover \eqref{eq:def spectral projection} induces a bounded functional calculus in the sense that
\begin{equation}\label{eq:functional_calculus_normal}
	\forall t\in\R,\quad  e^{itH}\mathds{1}_I(H):=\frac{1}{2\pi i}\wlim_{\varepsilon\rightarrow 0^+}\int_Ie^{it\lambda}\left(\Res_H(\lambda+i\varepsilon)-\Res_H(\lambda-i\varepsilon)\right)\mathrm{d}\lambda.
\end{equation}
In particular, $\lbrace e^{itH}\rbrace_{t\in\R}$ is uniformly bounded on the range of $\mathds{1}_I(H)$. Such spectral projections for non self-adjoint operators have been considered by several authors, see \cite{FaFr18_01} for dissipative operators, \cite{Go70} and \cite{Go71} for general non-self-adjoint operators,  \cite{KaYa76_09} for wave operators using a stationary approach and \cite{Mo66_11} and \cite{Mo67_09} for differential operators. 

If $I\subset\sigma_\mathrm{ess}(H)$ is a closed interval whose interior contains spectral singularities, then we define the `regularized spectral projection' associated to $H$ in $I$ by setting
\begin{equation}\label{eq:def regularized spectral projection}
	(h\mathds{1}_I)(H):=\frac{1}{2\pi i}\wlim_{\varepsilon\rightarrow 0^+}\int_I\left(h(\lambda+i\varepsilon)\Res_H(\lambda+i\varepsilon)-h(\lambda-i\varepsilon)\Res_H(\lambda-i\varepsilon)\right)\mathrm{d}\lambda.
\end{equation}
where $h$ is a function that regularizes the singularity of the resolvent at the spectral singularity. More precisely, assuming in particular Hypothesis \ref{hyp:spectral singularities}, $h$ is chosen as
\begin{equation}\label{eq:defh}
h(z):=\prod_{j\in\mathcal{S}_I}r_j(z) , \quad \mathcal{S}_I:=\{j\in\llbracket 1,n\rrbracket\cup\{\infty\} , \, \lambda_j \in I \text{ is a spectral singularity}\},
\end{equation}
where $r_j$ are defined in \eqref{eq:rj}--\eqref{eq:rinfty}. Similarly as above, under our assumptions, \eqref{eq:def regularized spectral projection} exists and defines a Borel functional calculus in the sense that
\begin{equation}\label{eq:reg_funct_calc}
	\forall t\in\R,\quad  e^{itH}(h\mathds{1}_I)(H):=\frac{1}{2\pi i}\wlim_{\varepsilon\rightarrow 0^+}\int_Ie^{it\lambda}\left(h(\lambda+i\varepsilon)\Res_H(\lambda+i\varepsilon)-h(\lambda-i\varepsilon)\Res_H(\lambda-i\varepsilon)\right)\mathrm{d}\lambda.
\end{equation}
See \cite{FaFr22_06}. In particular, $\lbrace e^{itH}\rbrace_{t\in\R}$ is uniformly bounded on the range of $(h\mathds{1}_I)(H)$.

Finally, we will use the following identity (proven in \cite{FaFr22_06}) which generalizes the well-known resolution of identity for self-adjoint operators to our context:
\begin{equation}\label{eq:resolution identity}
	r(H)=r(H)\Pi_\mathrm{disc}(H)+\frac{1}{2\pi i}\wlim_{\varepsilon\rightarrow \infty}\int_{\sigma_\mathrm{ess}(H)}\left(r(\lambda+i\varepsilon)\Res_H(\lambda+i\varepsilon)-r(\lambda-i\varepsilon)\Res_H(\lambda-i\varepsilon)\right)\mathrm{d}\lambda,
\end{equation}
where $r$ is the regularizing function of Hypothesis \ref{hyp:spectral singularities}.

\section{The Wave operators}\label{sec:Wave Operators}

\subsection{Preliminary results} 

We begin with preliminary results that will allow us to prove the existence and study the wave operators in the sequel. 

\begin{proposition}\label{prop:existence strong limit wave operator general 1}
	Suppose that Hypothesis \ref{hyp:LAP H0} holds. Let $A\in\mathcal{B}(\Hi)$ be such that there exists $c>0$ satisfying 
	\begin{equation*}
		\forall u\in\Hi, \quad \int_0^\infty\norme{Ce^{\pm itH^\star }A^\star u}_\Hi^2\mathrm{d}t\leq c\norme{u}^2_\Hi.
	\end{equation*}
Then the strong limits 
\begin{equation*}
	\slim_{t\rightarrow \pm \infty}Ae^{itH}e^{-itH_0} \quad \text{and }\quad \slim_{t\rightarrow \pm \infty}A^\star e^{ itH_0}e^{-itH^\star } 
\end{equation*}
exist. 
\end{proposition}

We can state a second version of the last proposition.

\begin{proposition}\label{prop:existence strong limit wave operator general 2}
	Suppose that Hypothesis \ref{hyp:LAP H0} holds. Let $A\in\mathcal{B}(\Hi)$ be such that there exists $c>0$ satisfying 
\begin{equation*}
	\forall u\in\Hi, \quad \int_0^\infty\norme{Ce^{\pm itH }A u}_\Hi^2\mathrm{d}t\leq c\norme{u}^2_\Hi.
\end{equation*}
Then the strong limits 
\begin{equation*}
	\slim_{t\rightarrow \pm \infty}e^{itH_0}e^{-itH}A \quad \text{and }\quad \slim_{t\rightarrow \pm \infty}Ae^{ itH^\star}e^{-itH_0 }
\end{equation*}
exist.
\end{proposition}
The proofs of Propositions \ref{prop:existence strong limit wave operator general 1} and \ref{prop:existence strong limit wave operator general 2} are standard and based on Cook's method (see e.g. \cite{ReSi80_01} or \cite{YA98}). For the convenience of the reader, we give a proof of Proposition \ref{prop:existence strong limit wave operator general 1} in Appendix \ref{App:existence and property wave operators}.

\subsection{The regularized wave operators $W_\pm(H,H_0)$ and $W_\pm(H^\star,H_0)$}\label{subsec:reg_wave}

We recall the following definition. Note that $\Pi_\mathrm{ac}(H)$ is the projection onto the absolutely continuous spectral subspace of $H$ defined in section \ref{subsec:orga}.
\begin{definition}\label{def:ref wave operators}
 The regularized wave operators associated to $H$ and $H_0$ (respectively to $H^\star$ and $H_0$) are defined by
    \begin{align*}
        W_\pm(H,H_0)&:=\slim_{t\rightarrow\pm\infty}e^{itH}\Pi_\mathrm{ac}(H)r_\mp(H)e^{-itH_0},\\
        W_\pm(H^\star,H_0)&:=\slim_{t\rightarrow\pm \infty}e^{itH^\star}\Pi_\mathrm{ac}(H)\bar{r}_\mp(H^\star)e^{-itH_0}.
    \end{align*}
\end{definition}

In this section we prove the existence and study the properties of these wave operators. However, since Definition \ref{def:ref wave operators} is not standard, we begin with a few remarks justifying it.

\subsubsection{Remarks on Definition \ref{def:ref wave operators} of the regularized wave operators}\label{subsubsec:Justification of the definition}

To fix the ideas, we consider here the wave operator $W_-(H,H_0)$. In the usual setting, when $H$ and $H_0$ are self-adjoint operators acting on the same Hilbert space and when $H_0$ has purely absolutely continuous spectrum, the common definition of the wave operators $W_-(H,H_0)$ is
\begin{equation}\label{eq:classical definition wave operator}
W_-(H,H_0):=\slim_{t\rightarrow-\infty}e^{itH}e^{-itH_0}.
\end{equation}
Assuming for simplicity that $H$ has finitely many eigenvalues with finite multiplicities and no singular continuous spectrum, we can write, for all $t\ge0$ and $u\in\Hi$,
\begin{equation*}
e^{itH}e^{-itH_0}u=e^{itH}\Pi_\mathrm{ac}(H)e^{-itH_0}u+e^{itH}\Pi_\mathrm{p}(H)e^{-itH_0}u,
\end{equation*}
and since $\Pi_\mathrm{p}(H)$ is compact and $H_0$ has purely absolutely continuous spectrum, it is well-known that  
\begin{equation}\label{eq:Pip}
\Pi_\mathrm{p}(H)e^{-itH_0}u\xrightarrow[t\rightarrow -\infty]{} 0
\end{equation}
(see e.g \cite[Lemma 1.4.1]{YA98}). Hence proving the existence of \eqref{eq:classical definition wave operator} is equivalent to proving the existence of 
\begin{equation*}
\slim_{t\rightarrow-\infty}e^{itH}\Pi_\mathrm{ac}(H)e^{-itH_0}.
\end{equation*}
In our context, however, the group $\{e^{itH}\}$ is not uniformly bounded ($\|e^{itH}\Pi_\mathrm{p}(H)\|_{\mathcal{B}(\Hi)}$ may even blow up exponentially as $t\rightarrow -\infty$). Therefore \eqref{eq:classical definition wave operator} fails in general. 

Still, a natural definition in our context might be 
\begin{equation}\label{eq:possible def}
W_-(H,H_0):=\slim_{t\rightarrow -\infty}e^{itH}\Pi_\mathrm{ac}(H)e^{-itH_0}.
\end{equation}
The next proposition however shows that if $H$ has an outgoing spectral singularity, then $C$ is not locally relatively smooth with respect to $H$, and hence the usual Cook criterion do not apply to prove the existence of \eqref{eq:possible def}. The purpose of adding a regularizing operator in the definition of $W_-(H,H_0)$ (see Definition \ref{def:ref wave operators}) is precisely to overcome this issue.

\begin{proposition}
Suppose that the assumptions \ref{hyp:LAP H0}-\ref{hyp:Conjugate operator} hold. Suppose there exists a compact interval of $J\subset\sigma_\mathrm{ess}(H)$ without eigenvalue of $H$ and $u\in\Hi$, such that
\begin{equation*}
\lim_{\varepsilon\rightarrow 0^+}\int_J\norme{C\Res_{H}(\lambda-i\varepsilon)CWu}^2_\Hi\mathrm{d}\lambda=\infty.
\end{equation*}
Then, $C$ is not $H$-smooth in the interval $J$. 
\end{proposition}

\begin{proof}
Let $v=CWu$. Parseval's theorem gives
\begin{equation*}
\int_0^\infty\norme{Ce^{itH}\Pi_\mathrm{ac}(H)v}_\Hi^2\mathrm{d}t=\frac{1}{2\pi}\lim_{\varepsilon\rightarrow 0^+}\int_{\sigma_\mathrm{ess}(H)}\norme{C\Res_{H}(\lambda-i\varepsilon)\Pi_\mathrm{ac}(H)v}_\Hi^2\mathrm{d}\lambda.
\end{equation*}
Fixing $\varepsilon>0$ small enough, we have 
\begin{align*}
&\int_{\sigma_\mathrm{ess}(H)}\norme{C\Res_{H}(\lambda-i\varepsilon)\Pi_\mathrm{ac}(H)v}^2_\Hi\mathrm{d}\lambda\geq \int_{J}\norme{C\Res_{H}(\lambda-i\varepsilon)\Pi_\mathrm{ac}(H)v}^2_\Hi\mathrm{d}\lambda\\
&\geq \frac{1}{2}\int_{J}\norme{C\Res_{H}(\lambda-i\varepsilon)v}^2_\Hi\mathrm{d}\lambda - \frac{1}{2}\int_{J}\norme{C\Res_{H}(\lambda-i\varepsilon)\Pi_\mathrm{p}(H)v}^2_\Hi\mathrm{d}\lambda.
\end{align*}
As $J$ does not contain any eigenvalue of $H$,
\begin{equation*}
\lim_{\varepsilon\rightarrow 0^+}\int_J\norme{C\Res_{H}(\lambda-i\varepsilon)\Pi_\mathrm{p}(H)v}^2\mathrm{d}\lambda
\end{equation*} exists in $\Hi$.
Moreover, by assumption
\begin{equation*}
\lim_{\varepsilon\rightarrow 0^+}\int_{J}\norme{C\Res_{H}(\lambda-i\varepsilon)v}^2_\Hi\mathrm{d}\lambda=\int_{J}\norme{C\Res_{H}(\lambda-i\varepsilon)CWu}^2_\Hi\mathrm{d}\lambda=\infty.
\end{equation*}
This prove that $C$ is not $H$-smooth. 
\end{proof}

Based on the previous proposition, we believe that, in general, the limit \eqref{eq:possible def} does not exist. Another argument supporting our Definition \ref{def:ref wave operators} come from the properties proven later in this section, namely that $W_\pm(H,H_0)$ exist, are injective and with ranges dense in $\Hi_{\mathrm{ac}}(H)$.

To conclude this subsection, we underline that we could also have defined the wave operators by using the regularizing operator $r(H)$ instead of $r_\pm(H)$ in Definition \ref{def:ref wave operators}. The advantage would have been to be able to rely on the functional calculus of \cite{FaFr22_06} to prove the existence of the wave operators. See the discussion about local wave operators in Section \ref{subsec:regularized wave operators} for related arguments proving the existence of wave operators using the functional calculus. The `inconvenience' of using $r(H)$ instead of $r_\pm(H)$ in Definition \ref{def:ref wave operators} is that `a part of' $r(H)$ is superfluous for the strong limit to exist, which may have the effect of somehow projecting out scattering states from the range of the wave operators.

We mention that such definitions of wave operators with an `identification' operator appearing between the two evolution groups has been considered by Kato in \cite{Kato67_05}, in a general setting where the evolution groups act on different Hilbert spaces. In the unitary case, various properties of the wave operators depending on the assumptions made on  the identification operator are given in \cite{YA98}. 

\subsubsection{Existence of $W_\pm(H,H_0)$ and $W_\pm(H^\star,H_0)$}
Now we turn to the proof of the existence of the regularized wave operators $W_\pm(H,H_0)$ and $W_\pm(H^\star,H_0)$. It suffices to apply Propositions \ref{prop:existence strong limit wave operator general 1} and  \ref{prop:existence strong limit wave operator general 2} with $A=r_\mp(H)\Pi_\mathrm{ac}(H)$. The following lemma shows that the conditions of Propositions \ref{prop:existence strong limit wave operator general 1} and  \ref{prop:existence strong limit wave operator general 2} are indeed satisfied.

\begin{lemma}\label{Cr Kato smooth H}
Suppose that Hypotheses \ref{hyp:LAP H0}-\ref{hyp:Conjugate operator} hold. Then there exists $c_0>0$ such that, for all $u\in\Hi$, 
\begin{equation}\label{eq:Cr Kato smooth H}
    \int_0^\infty\norme{Cr_\mp(H)e^{\pm itH}\Pi_\mathrm{ac}(H)u}_\Hi^2\mathrm{d}t\leq c_0\norme{u}_\Hi^2,
\end{equation}
    and 
\begin{equation}
\int_0^\infty\norme{C\bar{r}_\mp(H^\star)e^{\pm itH^\star}\Pi_\mathrm{ac}(H^\star) u}_\Hi^2\mathrm{d}t\leq c_0\norme{u}_\Hi^2.
\end{equation}\label{eq:Cr Kato smooth H star}
\end{lemma}

Note that if $H$ does not have any spectral singularity, then $r_\pm(H)=\mathrm{Id}$ and this lemma shows that $C$ is smooth with respect to $H$.

The proof of Lemma \ref{Cr Kato smooth H} is based on the following observation. If $\lambda_j$ is a spectral singularity of order $1$, then, for all $z\in\rho(H)\backslash\lbrace z_0\rbrace$, if $j\neq\infty$, 
\begin{align*}
r_j(H)-r_j(z)&=(H-\lambda_j)\Res_H(z_0)-(z-\lambda_j)(z-z_0)^{-1}\\
&=\left(H-z\right)\Res_H(z_0)(\lambda_j-z_0)(z-z_0)^{-1}.
\end{align*}
Likewise, if $j=\infty$,
\begin{align*}
r_\infty(H)-r_\infty(z)&=\Res_H(z_0)-(z-z_0)^{-1}\\
&=(H-z)\Res_H(z_0)(z-z_0)^{-1}.
\end{align*}
Next, if $\lambda_j$ is a spectral singularity of order $\mu_j$, then, if $j\neq\infty$, 
\begin{align}
r_j(H)-r_j(z)&=\left((H-\lambda_j)\Res_H(z_0)\right)^{\nu_j}-\left((\lambda_j-z)(\lambda_j-z_0)^{-1}\right)^{\nu_j}\nonumber\\
&=\left(\left(H-z\right)\Res_H(z_0)\left(\frac{\lambda_j-z_0}{z-z_0}\right)\right)^{\nu_j}\left(\sum_{k=0}^{\nu_j-1}\left((H-\lambda_j)\Res_H(z_0)\right)^{k}\left(\frac{z-\lambda_j}{z-z_0}\right)^{\nu_j-1-k}\right),\label{eq:resolvant formula generalized}
\end{align}
while if $j=\infty$,
\begin{align*}
r_\infty(H)-r_\infty(z)&=\left(\left(H-z\right)\Res_H(z_0)\left(z-z_0\right)^{-1}\right)^{\nu_\infty}\left(\sum_{k=0}^{\nu_\infty-1}\Res_H(z_0)^{k}\left(z-z_0\right)^{k+1-\nu_\infty}\right).
\end{align*}

Now we are ready to prove Lemma \ref{Cr Kato smooth H}.

\begin{proof}[Proof of Lemma \ref{Cr Kato smooth H}]
Let $u\in\Hi$. We will prove there exists $c>0$ such that
\begin{equation}\label{eq: proof lemma Cr_+ smooth kato}
\int_0^\infty\norme{Cr_+e^{-itH}\Pi_\mathrm{ac}(H)u}_\Hi^2\mathrm{d}t\leq c\norme{u}_\Hi^2.
\end{equation} 
The other inequalities can be obtained in the same way. Parseval's identity gives
\begin{equation*}
\int_0^\infty \norme{Cr_+(H)e^{-itH}\Pi_\mathrm{ac}(H)u}_\Hi^2\mathrm{d}t=\frac{1}{2\pi}\lim_{\varepsilon\rightarrow 0^+}\int_\R \norme{Cr_+(H)\Res_H(\lambda+i\varepsilon)\Pi_\mathrm{ac}(H)u}_\Hi^2\mathrm{d}\lambda
\end{equation*}
Let $n_+$ denote the number of incoming spectral singularities. We partition $\sigma_\mathrm{ess}(H)$ into a finite union of $n_++2$ interval $J_j$'s such that i) $J_0$ is an interval of the form $(-\infty,a]$ with $a<0$, $\abso{a}\geq \norme{V}_{\mathcal{B}(\Hi)}+1$, ii) for all $j\in\llbracket 1,n_+\rrbracket$, $J_j$ is a compact interval whose interior contains the incoming spectral singularity $\lambda_j$ and no other spectral singularities and iii) $J_{n_++1}$ is an interval of the form $[b,\infty)$, with $b$ large enough, containing a possible singularity at infinity. We then obtain that, for all $\varepsilon>0$ small enough,
\begin{align*}
\int_\R \norme{Cr_+(H)\Res_H(\lambda+i\varepsilon)\Pi_\mathrm{ac}(H)u}_\Hi^2\mathrm{d}\lambda=\sum_{j=0}^{n_++1}\int_{J_j}\norme{Cr_+(H)\Res_H(\lambda+i\varepsilon)\Pi_\mathrm{ac}(H)u}^2_\Hi\mathrm{d}\lambda.
\end{align*}
Therefore, to prove \eqref{eq: proof lemma Cr_+ smooth kato}, it suffices to prove that for all $j\in\llbracket 0,n_++1\rrbracket$, there exists $c_j>0$ such that 
\begin{equation}\label{eq:proof C smooth local estimate to prove}
\lim_{\varepsilon\rightarrow 0^+}\int_{J_j}\norme{Cr_+(H)\Res_H(\lambda+i\varepsilon)\Pi_\mathrm{ac}(H)u}_\Hi^2\mathrm{d}\lambda\leq c_j\norme{u}_\Hi^2. 
\end{equation}
First, we estimate the integral over $J_0$. As $H_0$ is self-adjoint and $J_0\subset (-\infty,0)$, we remark that for all $\lambda\in J_0$, for all $\varepsilon>0$ small enough,
\begin{equation*}
\norme{\Res_0(\lambda+i\varepsilon)}\leq \frac{1}{\mathrm{dist}(\lambda+i\varepsilon , \sigma(H_0))} \le  \frac{1}{\abso{\lambda+i\varepsilon}} .
\end{equation*}
Next, using the resolvent identity, we have 
\begin{align*}
\norme{\Res_H(\lambda+i\varepsilon)}_{\mathcal{B}(\Hi)}&\leq \norme{\Res_0(\lambda+i\varepsilon)}_{\mathcal{B}(\Hi)}+\norme{\Res_0(\lambda+i\varepsilon)V\Res_H(\lambda+i\varepsilon)}_{\mathcal{B}(\Hi)}\\
&\leq \frac{1}{\abso{\lambda+i\varepsilon}}+\norme{V}_{\mathcal{B}(\Hi)}\frac{1}{\abso{\lambda+i\varepsilon}}\norme{\Res_H(\lambda+i\varepsilon)}_{\mathcal{B}(\Hi)}\\
&\leq \frac{1}{\abso{\lambda+i\varepsilon}}+\norme{V}_{\mathcal{B}(\Hi)}\frac{1}{\abso{a}-\abso{\varepsilon}}\norme{\Res_H(\lambda+i\varepsilon)}_{\mathcal{B}(\Hi)}
\end{align*} 
As $\norme{V}_{\mathcal{B}(\Hi)}(\abso{a}-\abso{\varepsilon})^{-1}<1$ we have that for all $\lambda\in J_0$, for all $\varepsilon>0$ small enough,
\begin{equation*}
\norme{\Res_H(\lambda+i\varepsilon)}_{\mathcal{B}(\Hi)}\leq \frac{b}{\abso{\lambda+i\varepsilon}},
\end{equation*}
for some positive constant $b$. Therefore we deduce that there exists $c_0>0$ such that
\begin{equation*}
\int_{J_0}\norme{C\Res_H(\lambda+i\varepsilon)r_+(H)\Pi_\mathrm{ac}(H)u}^2\mathrm{d}\lambda\leq c_0\int_{J_0}\frac{1}{\abso{\lambda+i\varepsilon}^2}\mathrm{d}\lambda\norme{u}_\Hi^2. 
\end{equation*}
As $J_0$ is a closed interval not containing $0$, Lebesgue's dominated theorem shows that
\begin{equation*}
\lim_{\varepsilon\rightarrow 0^+}\int_{J_0}\frac{1}{\abso{\lambda+i\varepsilon}^2}\mathrm{d}\lambda
\end{equation*}
exists. This proves \eqref{eq:proof C smooth local estimate to prove} for $j=0$. 

Now let $j\in\llbracket 1,n_++1\rrbracket$. Writing 
\begin{equation*}
\tilde{r}_j(H):=\prod_{\substack{{k=1}\\k\neq j}}^{n_++1}r_k(H)\in\mathcal{B}(\Hi),
\end{equation*}
and setting $v=\tilde{r}_j(H)u$, we have 
\begin{align}
\int_{J_j}&\norme{Cr_+(H)\Res_H(\lambda+i\varepsilon)\Pi_\mathrm{ac}(H)u}_\Hi^2\mathrm{d}\lambda=\int_{J_j}\norme{Cr_j(H)\Res_H(\lambda+i\varepsilon)\Pi_\mathrm{ac}(H)v}_\Hi^2\mathrm{d}\lambda\nonumber\\
&\leq 2\int_{J_j}\norme{C\left(r_j(H)-r_j(\lambda+i\varepsilon)\right)\Res_H(\lambda+i\varepsilon)\Pi_\mathrm{ac}(H)v}_\Hi^2\mathrm{d}\lambda\label{eq:proof lemma Cr_+ smooth Kato eq1}\\
&+2\int_{J_j}\abso{r_j(\lambda+i\varepsilon)}^2\norme{C\Res_H(\lambda+i\varepsilon)\Pi_\mathrm{ac}(H)v}_\Hi^2\mathrm{d}\lambda.\label{eq:proof lemma Cr_+ smooth Kato eq2}
\end{align}
To estimate \eqref{eq:proof lemma Cr_+ smooth Kato eq1}, we use \eqref{eq:resolvant formula generalized}: for all $j\in\llbracket 1,n_+\rrbracket$, we have 
\begin{align*}
&\int_{J_j}\norme{C\left(r_j(H)-r_j(\lambda+i\varepsilon)\right)\Res_H(\lambda+i\varepsilon)\Pi_\mathrm{ac}(H)v}_\Hi^2\mathrm{d}\lambda\\
&=\int_{J_j}\norme{C\left(\Res_H(z_0)\left(\frac{\lambda_j-z_0}{\lambda+i\varepsilon-z_0}\right)\right)^{\nu_j}A(\lambda+i\varepsilon)(H-(\lambda+i\varepsilon))^{\nu_j-1}\Pi_\mathrm{ac}(H)v}_\Hi^2\mathrm{d}\lambda,
\end{align*}
where 
\begin{equation*}
A(\lambda+i\varepsilon):=\left(\sum_{k=0}^{\nu_j-1}\left((H-\lambda_j)\Res_H(z_0)\right)^{k}\left(\frac{\lambda+i\varepsilon-\lambda_j}{\lambda+i\varepsilon-z_0}\right)^{\nu_j-1-k}\right).
\end{equation*}
Clearly, for all $\lambda\in J_j$,
\begin{equation*}
\lim_{\varepsilon\rightarrow 0^+}C\left(\Res_H(z_0)\left(\frac{\lambda_j-z_0}{\lambda+i\varepsilon-z_0}\right)\right)^{\nu_j}A(\lambda+i\varepsilon)(H-(\lambda+i\varepsilon))^{\nu_j-1}\Pi_\mathrm{ac}(H)v,
\end{equation*}
exists in $\Hi$. We claim that there exists $\varepsilon_0>0$ such that 
\begin{equation}\label{eq:Cr smooth proof domi}
\sup_{\lambda\in J_j,\varepsilon<\varepsilon_0}\norme{C\left(\Res_H(z_0)\left(\frac{\lambda_j-z_0}{\lambda+i\varepsilon-z_0}\right)\right)^{\nu_j}A(\lambda+i\varepsilon)(H-(\lambda+i\varepsilon))^{\nu_j-1}\Pi_\mathrm{ac}(H)v}_\Hi<\infty.
\end{equation}
Indeed, for all $\varepsilon>0$ small enough it's clear that there exists $\varepsilon_1>0$ such that 
\begin{equation*}
	\sup_{\lambda\in J_j,0<\varepsilon<\varepsilon_1}\norme{A(\lambda+i\varepsilon)}_{\mathcal{B}(\Hi)}<\infty.
\end{equation*}
A direct computation shows that $\lambda\mapsto\Res_H(z_0)(H-\lambda+i\varepsilon)$ is uniformly bounded in $\mathcal{B}(\Hi)$ on $J_j$ with a bound independant of $\varepsilon$, thus there exists $b_j>0$ and $\varepsilon_2>0$ such that 
\begin{equation*}
	\sup_{\lambda\in J_j,0<\varepsilon<\varepsilon_2}\abso{\frac{\lambda_j-z_0}{\lambda+i\varepsilon-z_0}}^{\nu_j}\norme{C\Res_H(z_0)^{\nu_j}(H-(\lambda+i\varepsilon))^{\nu_j-1}\Pi_\mathrm{ac}(H)v}_\Hi<b_j\norme{v}_\Hi.
\end{equation*}
So we get \eqref{eq:Cr smooth proof domi}. Finally, using dominated convergence theorem, the map 
\begin{equation*}
	J_j\ni\lambda\mapsto C\left(\Res_H(z_0)\left(\frac{\lambda_j-z_0}{\lambda-z_0}\right)\right)^{\nu_j}A(\lambda)(H-\lambda)^{\nu_j-1}\Pi_\mathrm{ac}(H)v
\end{equation*}
is in $L^2(J_j,\Hi)$ and there exists $a_j>0$ such that
\begin{align}
&\int_{J_j}\norme{C\left(\Res_H(z_0)\left(\frac{\lambda_j-z_0}{\lambda-z_0}\right)\right)^{\nu_j}A(\lambda)(H-\lambda)^{\nu_j-1}\Pi_\mathrm{ac}(H)v}_\Hi^2\mathrm{d}\lambda\nonumber\\
&=\lim_{\varepsilon\rightarrow 0^+}\int_{J_j}\norme{C\left(\Res_H(z_0)\left(\frac{\lambda_j-z_0}{\lambda+i\varepsilon-z_0}\right)\right)^{\nu_j}A(\lambda+i\varepsilon)(H-(\lambda+i\varepsilon))^{\nu_j-1}\Pi_\mathrm{ac}(H)v}_\Hi^2\mathrm{d}\lambda\nonumber\\
&\leq  a_j\norme{u}_\Hi^2\label{eq:proof lemma Cr_+ smooth Kato estim1}.
\end{align}
where in the last inequality we used that $v=\tilde{r}_j(H)u$ and that $r_j(H)$ is bounded.\\
For the integral over $J_{n_++1}$ we have
\begin{align*}
&\int_{J_{n_++1}}\norme{C\left(r_\infty(H)-r_\infty(\lambda+i\varepsilon)\right)\Res_H(\lambda+i\varepsilon)\Pi_\mathrm{ac}(H)v}_\Hi^2\mathrm{d}\lambda\\
&=\int_{J_{n_++1}}\norme{CA(\lambda+i\varepsilon)\left(\Res_H(z_0)\left(\lambda+i\varepsilon-z_0\right)^{-1}\right)^{\nu_\infty}(H-(\lambda+i\varepsilon))^{\nu_\infty-1}\Pi_\mathrm{ac}(H)v}_\Hi^2\mathrm{d}\lambda,
\end{align*}
where 
\begin{equation*}
A(\lambda+i\varepsilon):=\left(\sum_{k=0}^{\nu_\infty-1}\Res_H(z_0)^{k}\left(\lambda+i\varepsilon-z_0\right)^{k+1-\nu_\infty}\right).
\end{equation*}
It is clear that for all $j\in J_{n_++1}$, 
\begin{equation*}
\lim_{\varepsilon\rightarrow 0^+}CA(\lambda+i\varepsilon)\left(\Res_H(z_0)\left(\lambda+i\varepsilon-z_0\right)^{-1}\right)^{\nu_\infty}(H-(\lambda+i\varepsilon))^{\nu_\infty-1}\Pi_\mathrm{ac}(H)v
\end{equation*}
exists in $\Hi$. Moreover, there exists $\varepsilon_1>0$ such that for all $\lambda\in J_{n_++1}$, for all $\varepsilon>0$ smaller than $\varepsilon_1$, 
\begin{equation*}
\norme{CA(\lambda+i\varepsilon)}_{\mathcal{B}(\Hi)}\leq \left(\nu_\infty\sup_{k\in\llbracket 0,\nu_\infty-1\rrbracket}\left(\abso{\lambda+i\varepsilon_1-z_0}^{k+1-\nu_\infty}\norme{C\Res_H(z_0)}^k_{\mathcal{B}(\Hi)}\right)\right)<\infty,
\end{equation*}
and a direct computation shows that 
\begin{equation*}
	\Res_H(z_0)(H-(\lambda+i\varepsilon))=\Id+\Res_H(z_0)(z_0-(\lambda+i\varepsilon)).
\end{equation*}
So for all $\varepsilon$ smaller than $\varepsilon_1$,
\begin{align*}
&\norme{\left(\Res_H(z_0)\left(\lambda+i\varepsilon-z_0\right)^{-1}\right)^{\nu_\infty}(H-(\lambda+i\varepsilon))^{\nu_\infty-1}\Pi_\mathrm{ac}(H)v}_\Hi\\
&\leq \abso{(\lambda+i\varepsilon_1-z_0)}^{-\nu_\infty}\norme{\Res_H(z_0)}_{\mathcal{B}(\Hi)}\left(1+\norme{\Res_H(z_0)}_{\mathcal{B}(\Hi)}\abso{z_0-(\lambda+i\varepsilon_1)}\right)^{\nu_\infty-1}\norme{\Pi_\mathrm{ac}(H)v}_\Hi\\
&=\mathcal{O}(\lambda^{-1}),\quad (\lambda\rightarrow\infty).
\end{align*}
Thus the map
\begin{equation*}
J_{n_++1}\ni\lambda\mapsto\abso{(\lambda+i\varepsilon_1-z_0)}^{-\nu_\infty}\norme{\Res_H(z_0)}_{\mathcal{B}(\Hi)}\left(1+\norme{\Res_H(z_0)}_{\mathcal{B}(\Hi)}\abso{z_0-(\lambda+i\varepsilon_1)}\right)^{\nu_\infty-1}\norme{\Pi_\mathrm{ac}(H)v}_\Hi,
\end{equation*}
is in $L^2(J_{n_++1},\Hi)$. Finally with Lebesgue's dominated convergence theorem, 
\begin{align*}
\lim_{\varepsilon\rightarrow 0^+}\int_{J_{n_++1}}\norme{C\left(r_\infty(H)-r_\infty(\lambda+i\varepsilon)\right)\Res_H(\lambda+i\varepsilon)\Pi_\mathrm{ac}(H)v}_\Hi^2\mathrm{d}\lambda
\end{align*}
exists and there exists $a_\infty>0$ such that 
\begin{align}
&\int_{J_{n_++1}}\norme{CA(\lambda)\left(\Res_H(z_0)\left(\lambda-z_0\right)^{-1}\right)^{\nu_\infty}(H-\lambda)^{\nu_\infty-1}\Pi_\mathrm{ac}(H)v}_\Hi^2\mathrm{d}\lambda\nonumber\\
&=\lim_{\varepsilon\rightarrow 0^+}\int_{J_{n_++1}}\norme{CA(\lambda+i\varepsilon)\left(\Res_H(z_0)\left(\lambda+i\varepsilon-z_0\right)^{-1}\right)^{\nu_\infty}(H-(\lambda+i\varepsilon))^{\nu_\infty-1}\Pi_\mathrm{ac}(H)v}_\Hi^2\mathrm{d}\lambda\nonumber\\
&\leq a_\infty\norme{u}^2_\Hi\label{eq:proof lemma Cr_+ smooth Kato estim2}.
\end{align}

It remains to estimate \eqref{eq:proof lemma Cr_+ smooth Kato eq2}. We use again the resolvent identity. For all $j\in\llbracket 1,n_++1\rrbracket$, and for all $\varepsilon>0$ small enough, we have
\begin{align*}
&\int_{J_j}\abso{r_j(\lambda+i\varepsilon)}^2\norme{C\Res_H(\lambda+i\varepsilon)\Pi_\mathrm{ac}(H)v}^2_\Hi\mathrm{d}\lambda\\
=&\int_{J_j}\abso{r_j(\lambda+i\varepsilon)}^2\norme{C\Res_0(\lambda+i\varepsilon)\Pi_\mathrm{ac}(H)v}^2_\Hi\mathrm{d}\lambda\\
&+\int_{J_j}\abso{r_j(\lambda+i\varepsilon)}^2\norme{C\Res_H(\lambda+i\varepsilon)CWC\Res_0(\lambda+i\varepsilon)\Pi_\mathrm{ac}(H)v}^2_\Hi\mathrm{d}\lambda.
\end{align*}
As for all $j\in\llbracket 1,n_++1\rrbracket$, 
\begin{equation*}
\sup_{\substack{\lambda\in J_j\\0<\varepsilon<\varepsilon_0}}\abso{r_j(\lambda+i\varepsilon)}\norme{C\Res_H(\lambda+i\varepsilon)CW}_{\mathcal{B}(\Hi)}<\infty
\end{equation*}
by Hypothesis \ref{hyp:spectral singularities}, and since $C$ is relatively smooth with respect to $H_0$, there exists $b_j>0$ such that 
\begin{align}
&\lim_{\varepsilon\rightarrow 0^+}\int_{J_j}\abso{r_j(\lambda+i\varepsilon)}^2\norme{C\Res_0(\lambda+i\varepsilon)\Pi_\mathrm{ac}(H)v}_\Hi^2\mathrm{d}\lambda\nonumber\\
&+\lim_{\varepsilon\rightarrow 0^+}\int_{J_j}\abso{r_j(\lambda+i\varepsilon)}^2\norme{C\Res_H(\lambda+i\varepsilon)CWC\Res_0(\lambda+i\varepsilon)\Pi_\mathrm{ac}(H)v}^2_\Hi\mathrm{d}\lambda\nonumber\\
&\leq b_j\norme{u}^2_\Hi.\label{eq:proof lemma Cr_+ smooth Kato estim3}
\end{align}
Combining \eqref{eq:proof lemma Cr_+ smooth Kato estim1},  \eqref{eq:proof lemma Cr_+ smooth Kato estim2} and \eqref{eq:proof lemma Cr_+ smooth Kato estim3}, we obtain \eqref{eq: proof lemma Cr_+ smooth kato}, which concludes the proof of the lemma.
\end{proof}

Applying Propositions \ref{prop:existence strong limit wave operator general 1} and \ref{prop:existence strong limit wave operator general 2}, the existence of the wave operators $W_\pm(H,H_0)$ then follows:

\begin{proposition}\label{prop:wave operator exist}
Suppose that Hypotheses \ref{hyp:LAP H0}-\ref{hyp:Conjugate operator} hold. Then the wave operators $W_\pm(H,H_0)$ and $W_\pm(H^\star,H_0)$ exist. 
\end{proposition}

\begin{proof}
It suffices to remark that $\Pi_\mathrm{ac}(H)r_\pm(H)$ commute with $H$ and then apply Lemma \ref{Cr Kato smooth H} together with Propositions \ref{prop:existence strong limit wave operator general 1} and \ref{prop:existence strong limit wave operator general 2}. 
\end{proof}

\subsubsection{Properties of the wave operators}

We derive in this section the main properties of the regularized wave operators. First we have that they satisfy the usual intertwining properties. 

\begin{proposition}\label{prop:interwining property}
Suppose that Hypotheses \ref{hyp:LAP H0}-\ref{hyp:Conjugate operator} hold.
Then for all $t\in\R$, for all $u\in\Hi$,
\begin{equation}\label{eq:interwining property group}
     e^{itH}W_\pm(H,H_0)u=W_\pm(H,H_0)e^{itH_0}u \quad \text{and}\quad e^{itH^\star}W_\pm(H^\star,H_0)u=W_\pm(H^\star,H_0)e^{itH_0}u.
\end{equation}
In particular, $W_\pm(H,H_0)\mathcal{D}(H_0)\subset\mathcal{D}(H_0)$,   $W_\pm(H^\star,H_0)\mathcal{D}(H_0)\subset\mathcal{D}(H_0)$ and for all $ u\in\mathcal{D}(H_0)$,
\begin{equation}\label{eq:interwining property operator}
     HW_\pm(H,H_0)u=W_\pm(H,H_0)H_0u\quad \text{and}\quad H^\star W_\pm(H^\star,H_0)u=W_\pm(H^\star,H_0)H_0u.
\end{equation}
\end{proposition}


The proof is standard. We will recall it in Appendix \ref{App:existence and property wave operators}. 

In the following proposition we show that the regularized wave operators are injective and that their ranges are dense in the absolutely continuous spectral subspace. 

\begin{proposition}\label{prop:range_and_ker_of_wave_operator}
Suppose Hypotheses \ref{hyp:LAP H0}-\ref{hyp:Conjugate operator} hold.
Then $W_\pm(H,H_0)$ and $W_\pm(H^\star,H_0)$ are injective and we have
\begin{equation*}
    \Ran(W_\pm(H,H_0))^\mathrm{cl}=\Hi_\mathrm{ac}(H)\quad \text{and}\quad \Ran(W_\pm(H^\star,H_0))^\mathrm{cl}=\Hi_\mathrm{ac}(H^\star).
\end{equation*}
\end{proposition}


\begin{proof}
	We show that $W_+(H,H_0)$ is injective. The proof for the other wave operators is similar. To show that $W_+(H,H_0)$ is injective, we compute its kernel. Let $u\in\Ker(W_+(H,H_0))$ then
	\begin{equation*}
		\lim_{t\rightarrow\infty}\norme{e^{itH}\Pi_\mathrm{ac}(H)r_-(H)e^{-itH_0}u}_\Hi=0.
	\end{equation*}
	Since $r(H)=r_-(H)\tilde r(H)$ for some bounded operator $\tilde r(H)$, we also have that
	\begin{equation*}
		\lim_{t\rightarrow\infty}\norme{e^{itH}\Pi_\mathrm{ac}(H)r(H)e^{-itH_0}u}_\Hi=0.
	\end{equation*}
	Next we claim that for all $t\in\R$, for all $v\in\Hi$, 
	\begin{equation}\label{eq:proof injectivity wave operator}
		\norme{e^{itH}\Pi_\mathrm{ac}(H)r(H)v}_\Hi\geq \norme{\Pi_\mathrm{ac}(H)r(H)^2v}_\Hi.
	\end{equation}
	Indeed, combining \eqref{eq:reg_funct_calc} and \eqref{eq:resolution identity} (see also \cite{FaFr22_06}) that, for all $t\in\R$,
	\begin{equation*}
		\norme{e^{itH}\Pi_\mathrm{ac}(H)r(H)v}_\Hi\leq \norme{\Pi_\mathrm{ac}(H)v}_\Hi.
	\end{equation*}
	Hence, for all $t\in\R$,
	\begin{equation}
		\norme{\Pi_\mathrm{ac}(H)r(H)v}_\Hi=\norme{e^{-itH}\Pi_\mathrm{ac}(H)r(H)e^{itH}v}_\Hi\leq\norme{e^{itH}\Pi_\mathrm{ac}(H)v}_\Hi.
	\end{equation}
	Applying the last inequality to $r(H)u$, we obtain \eqref{eq:proof injectivity wave operator}.
	
	Next, for all $t\in\R$, applying \eqref{eq:proof injectivity wave operator} to $v=e^{-itH_0}u$, we have
	\begin{equation*}
		\norme{e^{itH}\Pi_\mathrm{ac}(H)r(H)e^{-itH_0}u}_\Hi\geq \norme{\Pi_\mathrm{ac}(H)r(H)^2e^{-itH_0}u}_\Hi.
	\end{equation*}
	Then we write
	\begin{equation*}
		r(H)^2=\left(r(H)-r(H_0)\right)r(H)+r(H_0)\left(r(H)-r(H_0)\right)+r(H_0)^2
	\end{equation*}
	Proceeding by double induction over the number of spectral singularities and the order of each spectral singularity, it is not difficult to see that $r(H)-r(H_0)$ is a compact operator since  $C$ is relatively compact with respect to $H_0$. Hence $r(H)^2-r(H_0)^2$ is also compact. Using the triangular inequality, we obtain that 
	\begin{align*}
		&\norme{e^{itH}\Pi_\mathrm{ac}(H)r(H)e^{-itH_0}u}_\Hi\geq \norme{\Pi_\mathrm{ac}(H)r(H_0)^2e^{-itH_0}u}_\Hi\\
		&-\norme{\left(\Pi_\mathrm{ac}(H)\left( r(H)^2-r(H_0)^2 \right)\right)e^{-itH_0}u}_\Hi.
	\end{align*}
	As $r(H)^2-r(H_0)^2$ is compact, the second term in the right hand side of the last inequality vanishes when $t$ goes to $\infty$. Moreover, since $u$ is in the kernel of the wave operators, we have 
	\begin{equation*}
		0=\lim_{t\rightarrow\infty}\norme{\Pi_\mathrm{ac}(H)e^{-itH_0}r(H_0)^2u}_\Hi,
	\end{equation*}
	and since $\Pi_\mathrm{p}(H)$ is a compact operator, 
	\begin{equation*}
		0=\lim_{t\rightarrow\infty}\norme{\Pi_\mathrm{p}(H)e^{-itH_0}r(H_0)^2u}_\Hi.
	\end{equation*}
	Using $\Pi_\mathrm{p}(H)+\Pi_\mathrm{ac}(H)=\mathrm{Id}$, we then deduce from the previous estimates that 
	\begin{equation*}
		0=\lim_{t\rightarrow\infty}\norme{e^{-itH_0}r(H_0)^2u}_\Hi.
	\end{equation*}
	As the evolution group associated to $H_0$ is unitary and $r(H_0)^2$ is injective (since $H_0$ has no eigenvalue), we conclude that $u=0$, which proves that the wave operators are injective. 
	
	To show that $\Ran(W_+(H,H_0))^\mathrm{cl}=\Hi_\mathrm{ac}(H)$, it suffices to use that $\Ker((W_+(H,H_0))^\star)=\Hi_\mathrm{p}(H^\star)$ by Propositions \ref{prop:adjoint property wave operator} and \ref{prop:range/kernel wave operator bis} proven below.  
\end{proof}

Clearly, the last proposition shows that the regularized wave operators $W_\pm(H,H_0)$ are invertible from $\Hi$ to their range. In Section \ref{sec:Assymptotic completeness}, we will give conditions in order to have that $W_\pm(H,H_0)$ are invertible in $\mathcal{B}(\Hi,\Hi_{\mathrm{ac}}(H))$. 

\subsection{The regularized wave operators $W_\pm(H_0,H)$ and $W_\pm(H_0,H^\star)$}\label{subsec:autre operateur onde}

As shown by the proof of Proposition \ref{prop:range_and_ker_of_wave_operator}, an important ingredient to study the regularized wave operators $W_\pm(H,H_0)$ is to study their adjoint, in order to use the duality property between the kernel of an operator and the orthogonal of its range. Formally, the adjoint of $W_\pm(H,H_0)$ are given by $W_\pm(H_0,H^\star)$ where $W_\pm(H_0,H^\star)$ are defined as follows. 

\begin{definition}
    The wave operators associated to $H_0$ and $H$ (respectively $H_0$ and $H^\star$) are defined by
    \begin{align*}
        W_\pm(H_0,H)&:=\slim_{t\rightarrow \pm \infty}e^{itH_0}\Pi_\mathrm{ac}(H)r_\pm(H)e^{-itH},\\
        W_\pm(H_0,H^\star)&:=\slim_{t\rightarrow\pm\infty}e^{itH_0}\Pi_\mathrm{ac}(H^\star)\bar{r}_\pm(H^\star)e^{-itH^\star}.
    \end{align*}
\end{definition}

In the remainder of this section we state the existence and study the properties of these wave operators.

\subsubsection{Existence of $W_\pm(H_0,H)$ and $W_\pm(H_0,H^\star)$} 

The existence of $W_\pm(H_0,H)$ and $W_\pm(H_0,H^\star)$ follows in the same way as in the previous section.

\begin{proposition}\label{prop:existence wave operator adjoint}
Suppose that Hypotheses \ref{hyp:LAP H0}-\ref{hyp:Conjugate operator} hold. Then the wave operators $W_\pm(H_0,H)$ and $W_\pm(H_0,H^\star)$ exist.
\end{proposition}
\begin{proof}
It suffices to apply Propositions \ref{prop:existence strong limit wave operator general 1} and \ref{prop:existence strong limit wave operator general 2} together with Lemma \ref{Cr Kato smooth H}. 
\end{proof}

We mention the following interesting consequence of Proposition \ref{prop:existence wave operator adjoint}. Together with the Banach-Steinhaus Theorem, the last proposition shows that 
\begin{equation}\label{eq:group of H uniformly bounded }
\sup_{t\geq 0}\norme{e^{-itH}r_+(H)\Pi_\mathrm{ac}(H)}_{\mathcal{B}(\Hi)}<\infty\quad \text{and}\quad \sup_{t\geq 0}\norme{e^{itH}r_-(H)\Pi_\mathrm{ac}(H)}_{\mathcal{B}(\Hi)}<\infty,
\end{equation}
and 
\begin{equation}\label{eq: group of Hstar uniformly bounded}
\sup_{t\geq 0}\norme{e^{-itH^\star}\bar{r}_+(H^\star)\Pi_\mathrm{ac}(H^\star)}_{\mathcal{B}(\Hi)}<\infty\quad \text{and}\quad \sup_{t\geq 0}\norme{e^{itH^\star}\bar{r}_-(H^\star)\Pi_\mathrm{ac}(H^\star)}_{\mathcal{B}(\Hi)}<\infty.
\end{equation}
In other words, when restricted to $\mathrm{Ran}(r_\mp(H)\Pi_\mathrm{ac}(H))$, the semigroups generated by $\pm iH$ are uniformly bounded (roughly speaking $\Pi_\mathrm{ac}(H)$ projects out the generalized eigenstates of $H$, while $r_\mp(H)$ `projects out' the outgoing/incoming states corresponding to spectral singularities). This should be compared to the result that can be obtained using the functional calculus of \cite{FaFr22_06}, which only gives that the groups generated by $\pm iH$ are uniformly bounded when restricted to the range of $r(H)\Pi_\mathrm{ac}(H)$. 


\subsubsection{Properties of $W_\pm(H_0,H)$ and $W_\pm(H_0,H^\star)$}

The following intertwining properties follow in the same way as in the proof of Proposition \ref{prop:interwining property} (see Appendix \ref{App:existence and property wave operators}).

\begin{proposition}\label{prop:boundedness evolution group}
Suppose that Hypotheses \ref{hyp:LAP H0}-\ref{hyp:Conjugate operator} hold. Then for all $t\in\R$, for all $u\in\Hi$, 
\begin{equation*}
    e^{itH_0}W_\pm(H_0,H)=W_\pm(H_0,H)e^{itH} \quad \text{and}\quad e^{itH_0}W_\pm(H_0,H^\star)=W_\pm(H_0,H^\star)e^{itH}u.
\end{equation*}
In particular, $W_\pm(H_0,H)\mathcal{D}(H_0)\subset\mathcal{D}(H_0)$, $W_\pm(H_0,H^\star)\mathcal{D}(H_0)\subset\mathcal{D}(H_0)$ and for all $u\in\mathcal{D}(H_0)$, 
\begin{equation*}
    H_0W_\pm(H_0,H)u=W_\pm(H_0,H)Hu\quad \text{and}\quad H_0W_\pm(H_0,H^\star)u=W_\pm(H_0,H^\star)H^\star u. 
\end{equation*}
\end{proposition}


The next proposition relates the wave operators $W_\pm(H,H_0)$ to the adjoints of $W_\mp(H_0,H^\star)$.

\begin{proposition}\label{prop:adjoint property wave operator}
Suppose that Hypotheses \ref{hyp:LAP H0}-\ref{hyp:Conjugate operator} hold. Then
\begin{equation*}
    (W_\pm(H,H_0))^\star=W_\pm(H_0,H^\star) \quad \text{and}\quad (W_\pm(H^\star,H_0))^\star=W_\pm(H_0,H) 
\end{equation*}
\end{proposition}

\begin{proof}
	We prove that $(W_\pm(H,H_0))^\star=W_\pm(H_0,H^\star)$ the second equality is obtained in the same way. For $u,v\in\Hi$, we have
	\begin{align*}
		\scal{W_\pm(H,H_0)u}{v}_\Hi&=\lim_{t\rightarrow\pm \infty}\scal{e^{itH}\Pi_\mathrm{ac}(H)r_\mp(H)e^{-itH_0}u}{v}_\Hi\\
		&=\lim_{t\rightarrow \pm\infty}\scal{u}{e^{itH_0}\bar{r}_\pm(H^\star)\Pi_\mathrm{ac}(H^\star)e^{-itH^\star}v}_\Hi\\
		&=\scal{u}{W_\pm(H_0,H^\star)v}_\Hi.
	\end{align*}
	This proves the result. 
\end{proof}

It should be noted that Proposition \ref{prop:adjoint property wave operator} and Hypothesis \ref{hyp:Conjugate operator} yield
\begin{equation*}
    \norme{W_\pm(H,H_0)}_{\mathcal{B}(\Hi)}=\norme{W_\pm(H_0,H^\star)}_{\mathcal{B}(\Hi)}=\norme{W_\mp(H^\star,H_0)}_{\mathcal{B}(\Hi)}=\norme{W_\mp(H_0,H)}_{\mathcal{B}(\Hi)}.
\end{equation*}

In the next proposition we characterize the ranges and kernels or $W_\pm(H_0,H)$ and $W_\pm(H_0,H^\star)$. Its proof relies on the characterization of the set of asymptotically disappearing states of $H$, defined by
\begin{equation}\label{eq:defHads}
	\Hi_\mathrm{ads}^\pm(H):=\left\lbrace u\in\Hi,\lim_{t\rightarrow\pm \infty}\norme{e^{-itH}u}_\Hi=0\right\rbrace^\mathrm{cl}.
\end{equation} 
It has been proved in \cite{FaFr22_06} that, under our assumptions, $\Hi_\mathrm{ads}(H)$ coincide with the direct sum of the generalized eigenspaces associated to eigenvalues with negative/positive imaginary part, that is 
\begin{equation}\label{eq:defHads2}
	\Hi_\mathrm{ads}^\pm(H) = \Hi_\mathrm{p}^\pm(H):=\left\lbrace u\in\Ker(H-\lambda)^{\mathrm{m}_\lambda},\lambda\in\sigma_{\mathrm{p}}(H),\mp\im(\lambda)>0 \right\rbrace,
\end{equation}
where we recall that $\mathrm{m}_\lambda$ is the algebraic multiplicity of the eigenvalue $\lambda$. 

\begin{proposition}\label{prop:range/kernel wave operator bis}
Suppose that Hypotheses \ref{hyp:LAP H0}-\ref{hyp:Conjugate operator} hold. Then
\begin{equation*}
    \Ker(W_\pm(H_0,H))=\Hi_\mathrm{ac}(H^\star)^\perp=\Hi_\mathrm{p}(H), \quad \Ker(W_\pm(H_0,H^\star))=\Hi_\mathrm{ac}(H)^\perp=\Hi_\mathrm{p}(H^\star),
\end{equation*}
and
\begin{equation*}
    \Ran(W_\pm(H_0,H))^\mathrm{cl}=\Hi = \Ran(W_\pm(H_0,H^\star))^\mathrm{cl}.
\end{equation*}
\end{proposition}

\begin{proof}
	It follows from \cite{FaFr22_06} that $\Hi_\mathrm{ac}(H^\star)^\perp=\Hi_\mathrm{p}(H)$ and likewise for $\Hi_\mathrm{ac}(H)^\perp$. We prove that $\Ker(W_+(H_0,H))=\Hi_\mathrm{p}(H)$, the equalities $\Ker(W_-(H_0,H))=\Hi_\mathrm{p}(H)$ and $\Ker(W_\pm(H_0,H^\star))=\Hi_\mathrm{p}(H^\star)$ are proven in the same way.  It follows from the definition of $W_+(H_0,H)$ that $\Hi_\mathrm{p}(H)\subset\Ker(W_+(H_0,H))$. Hence we only need to prove the reverse inclusion. 
	
	Let $u\in\Ker(W_+(H_0,H))$. Using the unitarity of the group associated to $H_0$, we have that
	\begin{equation*}
		\lim_{t\rightarrow\infty}\norme{e^{- itH}r_+(H)\Pi_\mathrm{ac}(H)u}_\Hi=\lim_{t\rightarrow\infty}\norme{e^{ itH_0}r_+(H)\Pi_\mathrm{ac}(H)e^{- itH}u}_\Hi=0.
	\end{equation*}
	Thus $r_+(H)\Pi_\mathrm{ac}(H)u$ belongs to $\Hi_\mathrm{ads}^+(H)$. By the remark above, this implies that $r_+(H)\Pi_\mathrm{ac}(H)u\in\Hi_\mathrm{p}^+(H)\subset\Hi_\mathrm{p}(H)$. Together with the fact that $r_-(H)\Pi_\mathrm{ac}(H)u\in\Hi_\mathrm{ac}(H)$, this yields
	\begin{equation}
		r_-(H)\Pi_\mathrm{ac}(H)u=0.
	\end{equation}
	Finally as $r_-(H)$ is injective on $\Hi_\mathrm{ac}(H)$ (because the restriction of $H$ to $\Hi_\mathrm{ac}(H)$ has no point spectrum), we obtain that $\Pi_\mathrm{ac}(H)u=0$. This proves that $u\in\Hi_\mathrm{p}(H)$. 
	
	The fact that $\Ran(W_\pm(H_0,H))^\mathrm{cl}=\Hi$ follows from the facts that $W_\pm(H^\star,H_0)$ are injective by Proposition \ref{prop:range_and_ker_of_wave_operator} and that $(W_\pm(H_0,H))^\star=W_\pm(H^\star,H_0)$ by Proposition \ref{prop:adjoint property wave operator}.
\end{proof}

\subsection{The local regularized wave operators}\label{subsec:regularized wave operators}

It can be useful to define local wave operators on an interval $I$ of the essential spectrum. In this section we define the local regularized wave operators and we give their properties. 

We recall that given a closed in interval $I\subset\sigma_{\mathrm{ess}}(H_0)$, $h$ stands for the function defined in \eqref{eq:defh} which regularizes the spectral singularities of $H$ in $I$. It allows one to consider a regularized spectral projection $(h\mathds{1}_I)(H)$ as in \eqref{eq:def regularized spectral projection} which in turn induces a functional calculus \eqref{eq:reg_funct_calc}.

\begin{definition}
Let $I\subset\sigma_\mathrm{ess}(H)$ be a closed interval. We define the local regularized wave operators on the interval $I$ by 
\begin{equation*}
	W_\pm(H,H_0,I):=\slim_{t\rightarrow \pm\infty}e^{itH}(h\mathds{1}_I)(H)e^{-itH_0}\quad W_\pm(H^\star,H_0,I):=\slim_{t\rightarrow \pm\infty} e^{itH^\star}(\bar{h}\mathds{1}_I)(H^\star)e^{-itH_0}. 
\end{equation*}
\end{definition}

Note that in the particular case where $H$ does not have any spectral singularity in $I$, then the local wave operators associated to $H$ and $H_0$ reduce to
\begin{equation*}
		W_\pm(H,H_0,I):=\slim_{t\rightarrow \pm\infty}e^{itH}\mathds{1}_I(H)e^{-itH_0}\quad W_\pm(H^\star,H_0,I):=\slim_{t\rightarrow \pm\infty} e^{itH^\star}\mathds{1}_I(H^\star)e^{-itH_0}. 
\end{equation*}
Since $\mathds{1}_I(H)$ is a projection, this implies that in this case, $\Ran(W_\pm(H,H_0,I))\subset\Ran(\mathds{1}_I(H))$. 

\subsubsection{Existence of the local regularized wave operator}

As in Section \ref{subsec:reg_wave} for the regularized wave operators $W_\pm(H,H_0)$, the proof of the existence of the local regularized wave operators $W_\pm(H,H_0,I)$, is based on Proposition \ref{prop:existence strong limit wave operator general 1}. In order to apply this proposition, we need to prove that $C(h\mathds{1}_I)(H)$ is $H_0$-smooth, which is the purpose on the following lemma. Note that the proof, based on the functional calculus \eqref{eq:reg_funct_calc}, follows a different route  from that of Lemma \ref{Cr Kato smooth H}. The disadvantage of functional calculus is that we need to regularize both outgoing and incoming spectral singularities.
 
\begin{lemma}\label{lemma:Cr1_H H smooth}
	Suppose that Hypotheses \ref{hyp:LAP H0}-\ref{hyp:Conjugate operator} hold. Let $I\subset\sigma_\mathrm{ess}(H)$ be a closed interval. Then $C(h\mathds{1}_I)(H)$ is relatively smooth with respect to $H$ in the sense of Kato and $C(h\mathds{1}_{H^\star})(I)$ is relatively smooth with respect to $H^\star$. 
\end{lemma}


\begin{proof}
	It suffices to show that there exists $c>0$ such that, for all $u\in\Hi$, 
	\begin{equation}\label{eq:proof kato smooth local estimate to show}
		\int_0^\infty\norme{Ce^{itH}(h\mathds{1}_I)(H)u}_\Hi^2\mathrm{d}t\leq c\norme{u}_\Hi^2,
	\end{equation}
	the proof of  
	\begin{equation*}
		\int_0^\infty \norme{Ce^{-itH}(h\mathds{1}_I)(H)u}_\Hi^2\mathrm{d}t\leq c\norme{u}_\Hi^2
	\end{equation*}
	being similar. 
	Let $u,v\in\Hi$. Using the integral representation \eqref{eq:def regularized spectral projection} of $(h\mathds{1}_I)(H)$, for all $t\in\R$, we have 
	\begin{align}\label{eq:local kato smooth proof}
		&\scal{v}{Ce^{itH}(h\mathds{1}_I)(H)u}_\Hi=\scal{v}{Ce^{itH_0}h(H_0)\mathds{1}_{I}(H_0)u}_\Hi\\
		&-\frac{1}{2\pi i}\int_{I}e^{ it\lambda}h(\lambda)\scal{C\Res_0(\lambda\mp i0^+)Cv}{WC\Res_0(\lambda\pm i0^+)u}_\Hi\mathrm{d}\lambda\nonumber\\
		&+\frac{1}{2\pi i}\int_{I}e^{ it\lambda}h(\lambda)\scal{C\Res_{H^\star}(\lambda\mp i0^+)CW^\star C\Res_0(\lambda \mp i0^+)Cv}{WC\Res_0(\lambda\pm i0^+)u}_\Hi\mathrm{d}\lambda.\nonumber
	\end{align}
	Let $g_1^\pm$ and $g_2^\pm$ be the functions defined by
	\begin{align*}
		\R\ni\lambda\mapsto g_1^\pm(\lambda):&=h(\lambda)\scal{W^\star C\Res_0(\lambda\mp i0^+)Cv}{C\Res_0(\lambda\pm i0^+)u}_\Hi\mathds{1}_I(\lambda)\\
		&=h(\lambda)\scal{v}{C\Res_0(\lambda\pm i0^+)CWC\Res_0(\lambda\pm i0^+)u}_\Hi\mathds{1}_I(\lambda),
	\end{align*}
	and 
	\begin{align*}
		\R\ni\lambda\mapsto g_2^\pm(\lambda):&=h(\lambda)\scal{W^\star C\Res_{H^\star}(\lambda\mp i0^+)CW^\star C\Res_0(\lambda \mp i0^+)Cv}{C\Res_0(\lambda\mp i0^+)u}_\Hi\mathds{1}_I(\lambda)\\
		&=h(\lambda)\scal{v}{C\Res_0(\lambda\pm i0^+)CWC\Res_H(\lambda\pm i0^+)CWC\Res_0(\lambda\mp i0^+)u}_\Hi\mathds{1}_I(\lambda).
	\end{align*}
	We claim that $g_1^\pm$ and $g_2^\pm$ are in $L^1(\R,\C)\cap L^2(\R,\C)$. Indeed for all $\lambda\in I$, by Hypotheses \ref{hyp:LAP H0} and \ref{hyp:spectral singularities} and the Cauchy-Schwarz inequality, we have 
	\begin{equation}\label{eq:estimg1}
		\abso{g_1^\pm(\lambda)}^2\leq \norme{h}^2_\infty\sup_{z\in\C^\pm}\norme{C\Res_0(z)C}_{\mathcal{B}(\Hi)}^2\norme{W^\star}_{\mathcal{B}(\Hi)}\norme{C\Res_0(\lambda\pm i0^+)u}_\Hi^2\norme{v}_\Hi^2,
	\end{equation}
	and
	\begin{align}
		\abso{g_2^\pm(\lambda)}^2\leq&\left( \sup_{\substack{\re(z)\in I\\ \abso{\im(z)}\leq \varepsilon_0}}\left(\abso{h(z)}\norme{C\Res_H(z)CW\star}_{\mathcal{B}(\Hi)}\right)^2\right) \left(\sup_{z\in\C^\pm}\norme{C\Res_0(z)C}_{\mathcal{B}(\Hi)}^2\right)\nonumber\\
		&\times \norme{W}_{\mathcal{B}(\Hi)}^2\norme{C\Res_0(\lambda\pm i0^+)u}_\Hi^2\norme{v}^2_\Hi. \label{eq:estimg2}
	\end{align}
	Since $C$ is relatively smooth with respect to $H_0$, the function  $\lambda\mapsto C\Res_0(\lambda \pm i0^+)u$ is square integrable on  $I$.
	Moreover, the Cauchy-Schwarz inequality, there exists $c_1>0$ such that 
	\begin{align*}
		\int_\R\abso{g_1^\pm(\lambda)}\mathrm{d}\lambda&\leq \left(\int_I\abso{h(z)}\norme{W^\star C\Res_0(\lambda\mp i0^+)Cv}^2_\Hi\mathrm{d}\lambda\right)^{\frac{1}{2}} \left( \int_I\norme{C\Res_0(\lambda\pm i0^+)u}^2_\Hi\mathrm{d}\lambda\right)^{\frac{1}{2}}\\
		&\leq c_1\norme{u}_\Hi\norme{v}_\Hi.
	\end{align*}
	In the same way, we have 
	\begin{align*}
		\int_\R\abso{g_2(\lambda)}\mathrm{d}\lambda\leq c_2\norme{u}_\Hi\norme{v}_\Hi. 
	\end{align*}
	
	Using \eqref{eq:local kato smooth proof}, we deduce that there exists $c_3>0$ such that 
	\begin{align*}
		\int_0^\infty \abso{\scal{v}{Ce^{itH}(h\mathds{1}_I)(H)u}_\Hi}^2\mathrm{d}t&\leq c_3\bigg(\int_\R \abso{\scal{v}{Ce^{itH_0}\mathds{1}_{I}(H_0)u}_\Hi}^2\mathrm{d}t\\
		&+\int_\R\abso{\check{g}_1^\pm(t)}^2\mathrm{d}t+\int_\R\abso{\check{g}_2^\pm(t)}^2\mathrm{d}t\bigg),
	\end{align*}
	where for $i\in\lbrace 1,2\rbrace$,  $\check{g}_i^\pm$ is the inverse Fourier transform of $g_i^\pm$, i.e.
	\begin{equation*}
		\check{g}_i^\pm(t):=\int_\R e^{it\lambda}g_i^\pm(\lambda)\mathrm{d}\lambda
	\end{equation*}
	As $g_1^\pm$ and $g_2^\pm$ are in $L^2(\R,\C)$, Plancherel's equality implies that there exists $c_4>0$ such that 
	\begin{align*}
		\int_0^\infty \abso{\scal{v}{Ce^{itH}(h\mathds{1}_I)(H)u}_\Hi}^2\mathrm{d}t\leq c_4\bigg( &\int_\R \abso{\scal{v}{Ce^{itH_0}h(H_0)\mathds{1}_{I}(H_0)u}_\Hi}^2\mathrm{d}t\\
		&+\int_\R\abso{\check{g}_1^\pm(\lambda)}^2\mathrm{d}\lambda+\int_\R\abso{\check{g}_2^\pm(\lambda)}^2\mathrm{d}\lambda\bigg).
	\end{align*}
	Inserting \eqref{eq:estimg1}--\eqref{eq:estimg2} into the last inequality and using that $C$ is relatively smooth with respect to $H_0$, we obtain
	\begin{align*}
		\int_0^\infty \abso{\scal{v}{Ce^{itH}(h\mathds{1}_I)(H)u}_\Hi}^2\mathrm{d}t\leq c_5\|u\|_\Hi^2\|v\|_\Hi^2.
		\end{align*}
	Using $\norme{Ce^{itH}(h\mathds{1}_I)(H)u}_\Hi=\lim_{n\to\infty}\scal{v_n}{Ce^{itH}(h\mathds{1}_I)(H)u}_\Hi$ for some normalized sequence $(v_n)_{n\in\N}$ together with Fatou's lemma, this leads to
\eqref{eq:proof kato smooth local estimate to show}.
\end{proof} 

Using this lemma we obtain the existence of the local regularized wave operators.

\begin{proposition}
	Suppose that Hypotheses \ref{hyp:LAP H0}-\ref{hyp:Conjugate operator} hold. Let $I\subset\sigma_\mathrm{ess}(H)$ be a closed interval. Then $W_\pm(H,H_0,I)$ and $W_\pm(H^\star,H_0,I)$ exist.
\end{proposition}
\begin{proof}
It suffices to apply Lemma \ref{lemma:Cr1_H H smooth} together with Propositions \ref{prop:existence strong limit wave operator general 1} and \ref{prop:existence strong limit wave operator general 2}.
\end{proof}

Similarly as before, the adjoints of $W_\mp(H^\star,H_0,I)$ and $W_\mp(H,H_0,I)$ are formally given by
\begin{equation*}
	W_\pm(H_0,H,I):=\slim_{t\rightarrow\pm\infty}e^{itH_0}(h\mathds{1}_I)(H)e^{-itH}\quad 	W_\pm(H_0,H^\star,I):=\slim_{t\rightarrow\pm\infty}
e^{itH_0}(\bar{h}\mathds{1}_I)(H^\star)e^{-itH^\star}.
\end{equation*}
We then have the following existence result.

\begin{proposition}\label{prop:adjoint_local}
	Suppose that Hypotheses \ref{hyp:LAP H0}-\ref{hyp:Conjugate operator} hold. Let $I\subset\sigma_\mathrm{ess}(H)$ be a closed interval. Then the wave operators $W_\pm(H_0,H,I)$ and $W_\pm(H_0,H^\star,I)$ exist and satisfy
	\begin{equation*}
	W_\pm(H_0,H,I)^\star=W_\mp(H^\star,H_0,I), \quad W_\pm(H_0,H^\star,I)^\star=W_\mp(H,H_0,I).
	\end{equation*}
\end{proposition}
\begin{proof}
Again, it suffices to apply Lemma \ref{lemma:Cr1_H H smooth} together with Propositions \ref{prop:existence strong limit wave operator general 1} and \ref{prop:existence strong limit wave operator general 2} to obtain the existence. The adjoint properties follow from a direct computation.
\end{proof}

\subsubsection{Properties of the local regularized wave operators}

In this section we reformulate the properties of regularized wave operators stated in Section \ref{subsec:reg_wave} for the local regularized wave operators. First we have the following intertwining properties whose proofs, left to the reader, rely on the same arguments as those used in  the proof of Proposition \ref{prop:interwining property}.

\begin{proposition}\label{prop:interwining property local}
	Suppose that Hypotheses \ref{hyp:LAP H0}-\ref{hyp:Conjugate operator} hold.
	Then for all $t\in\R$, for all $u\in\Hi$,
	\begin{equation*}
		e^{itH}W_\pm(H,H_0,I)u=W_\pm(H,H_0,I)e^{itH_0}u \quad \text{and}\quad e^{itH^\star}W_\pm(H^\star,H_0,I)u=W_\pm(H^\star,H_0,I)e^{itH_0}u.
	\end{equation*}
	In particular, $W_\pm(H,H_0,I)\mathcal{D}(H_0)\subset\mathcal{D}(H_0)$,  $W_\pm(H^\star,H_0,I)\mathcal{D}(H_0)\subset\mathcal{D}(H_0)$ and for all $ u\in\mathcal{D}(H_0)$,
	\begin{equation*}
		HW_\pm(H,H_0,I)u=W_\pm(H,H_0,I)H_0u\quad \text{and}\quad H^\star W_\pm(H^\star,H_0,I)u=W_\pm(H^\star,H_0,I)H_0u.
	\end{equation*}
	Moreover, 
	\begin{align}\label{eq:interwining local spec proj}
		&(h\mathds{1}_I)(H)W_\pm(H,H_0,I)u=W\pm(H,H_0,I)h(H_0)\mathds{1}_I(H_0)u\\
		&(h\mathds{1}_I)(H^\star)W_\pm(H^\star,H_0,I)u=W_\pm(H^\star,H_0,I)h(H_0)\mathds{1}_I(H_0)u
	\end{align}
\end{proposition}

	Note that \eqref{eq:interwining local spec proj} relies on the existence of the spectral projection and the interwining property for the resolvant of $H$ and $H_0$, where the interwining property for the resolvent of $H$ and $H_0$ is proved by using the Laplace transform of the resolvent.


The next two propositions characterize the kernels and ranges of $W_\pm(H,H_0,I)$ and $W_\pm(H_0,H,I)$. 
The proof of Proposition \ref{prop:kernel and range local wave operator} relies on the same idea as that of Proposition \ref{prop:range_and_ker_of_wave_operator}, up to some technicalities due to the use of the regularized functional calculus \eqref{eq:reg_funct_calc}. Recall that $h$ stands for the function defined in \eqref{eq:defh} which regularizes the spectral singularities of $H$ in $I$.

\begin{proposition}\label{prop:kernel and range local wave operator}
		Suppose that Hypotheses \ref{hyp:LAP H0}-\ref{hyp:Conjugate operator} hold. Let $I\subset\sigma_\mathrm{ess}(H)$ be a closed interval. Then 
		\begin{equation*}
			\Ker(W_\pm(H,H_0,I))=\Ker(\mathds{1}_I(H_0)) \quad \text{and} \quad \Ker(W_\pm(H^\star,H_0,I))=\Ker(\mathds{1}_I(H_0)).
		\end{equation*}
		 Moreover,
		\begin{equation*}
			\Ran(W_\pm(H,H_0,I))^\mathrm{cl}= \Ran((r\mathds{1}_I)(H))^\mathrm{cl},\quad \text{and}\quad \Ran(W_\pm(H,H^\star_0,I))^\mathrm{cl}= \Ran((r\mathds{1}_I)(H^\star))^\mathrm{cl}.
		\end{equation*} 
\end{proposition}
\begin{proof}
	We show that $\Ker(W_+(H,H_0,I))=\Ker(\mathds{1}_I(H_0))$, the proof the other equalities characterizing the kernels are similar. Let $u\in \Ker(W_+(H,H_0,I))$, then 
	\begin{equation*}
		\lim_{t\rightarrow \infty} \norme{e^{itH}(h\mathds{1}_I)(H)e^{-itH_0}u}_\Hi=0. 
	\end{equation*}
	Composing by $(r\mathds{1}_I)(H)$ and using the intertwining properties of Proposition \ref{prop:interwining property local}, we obtain
	\begin{equation*}
		\lim_{t\rightarrow \infty}\norme{e^{itH}(h\mathds{1}_I)(H)e^{-itH_0}\mathds{1}_I(H_0)r(H_0)u}_\Hi=0. 
	\end{equation*}
	Next since $r=h\tilde r$ for some rational function $\tilde r$, we also have
	\begin{equation*}
		\lim_{t\rightarrow \infty} \norme{e^{itH}(r\mathds{1}_I)(H)e^{-itH_0}\mathds{1}_I(H_0)r(H_0)u}_\Hi=0. 
	\end{equation*}
	We claim that for all $t\in\R$, for all $v\in\Hi$, 
	\begin{equation}\label{eq:inequality proof kernel local regularized wave operator}
		\norme{e^{itH}(h\mathds{1}_I)(H)e^{-itH_0}v}_\Hi\geq \norme{(r\mathds{1}_I)(H)^2e^{-itH_0}v}_\Hi.
	\end{equation}
	Indeed, for all $t\in\R$, by the functional calculus \eqref{eq:reg_funct_calc} $\lbrace e^{itH}\rbrace_{t\in\R}$ is uniformly bounded on the range of $h(H)$. Thus we have 
	\begin{equation*}
		\norme{e^{-itH}(r\mathds{1}_I)(H)v}_\Hi\leq\norme{v}_\Hi.
	\end{equation*}
	Thus for all $t\in\R$, 
	\begin{equation*}
		\norme{(r\mathds{1}_I)(H)v}_\Hi=\norme{e^{-itH}(r\mathds{1}_I)(H)e^{itH}v}_\Hi\leq\norme{e^{itH}v}_\Hi.
	\end{equation*}	
	Applying the last inequality to $(r\mathds{1}_I)(H)e^{itH_0}v$, we obtain \eqref{eq:inequality proof kernel local regularized wave operator}. Now \eqref{eq:inequality proof kernel local regularized wave operator} yields
	\begin{equation*}
		\lim_{t\rightarrow\infty}\norme{(r\mathds{1}_I)(H)^2e^{-itH_0}\mathds{1}_I(H_0)r(H_0)u}_\Hi=0.
	\end{equation*}
	and hence
	\begin{equation}\label{eq:tech0}
		\lim_{t\rightarrow\infty}\norme{(r\mathds{1}_I)(H)^3e^{-itH_0}\mathds{1}_I(H_0)r(H_0)u}_\Hi=0.
	\end{equation}
	
	Now let $I^\mathrm{c} := (\sigma_\mathrm{ess}(H)\setminus I)^{\mathrm{cl}}$. We claim that 
	\begin{equation}\label{eq:proof inequality proof kernel local 2}
		\lim_{t\rightarrow\infty}\norme{(r\mathds{1}_{I^\mathrm{c}})(H)^3e^{-itH_0}\mathds{1}_I(H_0)r(H_0)u}_\Hi=0.
	\end{equation}
	Indeed, it follows again from the functional calculus \eqref{eq:reg_funct_calc} that there exists $c>0$ such that, for all $t\in\mathbb{R}$,
	\begin{equation*}
		\norme{(r\mathds{1}_{I^\mathrm{c}})(H)^3e^{-itH_0}\mathds{1}_I(H_0)r(H_0)u}_\Hi \leq c\norme{(r\mathds{1}_{I^\mathrm{c}})(H)e^{-itH}(r\mathds{1}_{I^\mathrm{c}})(H)e^{-itH_0}\mathds{1}_{I}(H_0)r(H_0)u}_\Hi.
	\end{equation*}
	Letting $t$ go to $\infty$ yields
	\begin{equation*}
		\limsup_{t\rightarrow\infty}\norme{(r\mathds{1}_{I^\mathrm{c}})(H)^3e^{-itH_0}\mathds{1}_I(H_0)r(H_0)u}_\Hi \leq c\norme{(r\mathds{1}_{I^\mathrm{c}})(H)W_+(H,H_0,I^\mathrm{c})\mathds{1}_{I}(H_0)r(H_0)u}_\Hi.
	\end{equation*}
	Finally using the intertwining properties of Proposition \ref{prop:interwining property local} together with the facts that $\mathds{1}_{I\cap I^\mathrm{c}}(H_0)=0$ since $H_0$ has purely absolutely continuous spectrum, we obtain \eqref{eq:proof inequality proof kernel local 2}. Proceeding in the same way, we have that
	\begin{equation*}
		\limsup_{t\rightarrow\infty}\norme{(r\mathds{1}_{I})(H)(r\mathds{1}_{I^\mathrm{c}})(H)^2e^{-itH_0}\mathds{1}_I(H_0)r(H_0)u}_\Hi=0.
	\end{equation*}
	and 
	\begin{equation*}
		\limsup_{t\rightarrow\infty}\norme{(r\mathds{1}_{I^\mathrm{c}})(H)(r\mathds{1}_{I})(H)^2e^{-itH_0}\mathds{1}_I(H_0)r(H_0)u}_\Hi=0.
	\end{equation*}
	Summing the last two previous limits with \eqref{eq:tech0} and \eqref{eq:proof inequality proof kernel local 2}  gives
	\begin{equation}\label{eq:tech 1}
		\limsup_{t\rightarrow\infty} \norme{(r\mathds{1}_\mathrm{ess})(H)^3e^{-itH_0}r(H_0)\mathds{1}_I(H_0)u}_\Hi=0. 
	\end{equation}
	
	Now $r(H)^3\Pi_\mathrm{p}(H)$ is compact under our assumptions and therefore
	\begin{equation}\label{eq:tech 2}
		\lim_{t\rightarrow \infty} \norme{r(H)^3\Pi_\mathrm{p}(H)e^{-itH_0}\mathds{1}_I(H_0)r(H_0)u}_\Hi=0.
	\end{equation}
	Summing \eqref{eq:tech 1} and \eqref{eq:tech 2}, it follows from the spectral resolution formula \eqref{eq:resolution identity} and the triangular inequality that
	\begin{equation*}
		\limsup_{t\rightarrow\infty}\norme{r(H)^3e^{-itH_0}r(H_0)\mathds{1}_I(H_0)u}_\Hi=0. 
	\end{equation*}
	Writing
	\begin{equation*}
		r(H)^3=(r(H)^2(r(H)-r(H_0))+r(H)(r(H)-r(H_0))r(H_0)+(r(H)-r(H_0))r(H_0)^2+r(H_0)^3,
	\end{equation*}
	and using that $r(H)-r(H_0)$ is compact, we obtain
	\begin{equation*}
		\lim_{t\rightarrow\infty}\norme{(r(H)^3-r(H_0)^3)e^{-itH_0}\mathds{1}_I(H_0)r(H_0)u}_\Hi=0.
	\end{equation*}
	and hence
	\begin{equation*}
		\lim_{t\rightarrow\infty}\norme{e^{-itH_0}\mathds{1}_I(H_0)r^4(H_0)u}_\Hi=0.
	\end{equation*}
	Since $r(H_0)^4$ is injective (because $H_0$ has no eigenvalue) and $\lbrace e^{-itH_0}\rbrace_{t\in\R}$ is unitary, we conclude that $u\in \Ker(\mathds{1}_I(H_0))$.
	
	Conversely, let $u\in\Ker(\mathds{1}_I(H_0))$. Then 
	\begin{equation*}
		W_+(H,H_0,I)(r\mathds{1}_I)(H_0)u=0
	\end{equation*}
	and using the intertwining properties of Proposition \ref{prop:interwining property local}, we also have
	\begin{equation}\label{eq:a1}
		(r\mathds{1}_I)(H)W_+(H,H_0,I)u=0. 
	\end{equation}
	Moreover, using that 
	\begin{equation*}
		(r\mathds{1}_{I^\mathrm{c}})(H)(r\mathds{1}_I)(H)=0,
	\end{equation*} 
	as follows from the functional calculus \eqref{eq:reg_funct_calc}, where $I^\mathrm{c} := (\sigma_\mathrm{ess}(H)\setminus I)^{\mathrm{cl}}$, we obtain
	\begin{equation}\label{eq:a2}
		(r\mathds{1}_{I^\mathrm{c}})(H)W_+(H,H_0,I)u=0.
	\end{equation}
	
	Now, we claim that 
	\begin{equation}\label{eq:proof kernel local wave operator recip}
		r(H)\Pi_\mathrm{p}(H)W_+(H,H_0,I)u=0. 
	\end{equation}
	Indeed, it suffices to prove that $W_+(H,H_0,I)u\in\Hi_\mathrm{ac}(H)$, where we recall that $\Hi_\mathrm{ac}(H)$ is defined as the closure of $\mathcal{M}(H)$, see \eqref{eq:def_M(H)}.
	Since $H_0$ has purely absolutely continuous spectrum, there exists a sequence $(u_n)_{n\in\N}\subset\mathcal{M}(H_0)$ such that $(u_n)_{n\in\N}$ converges to $u$ when $n\rightarrow \infty$. Using the intertwining properties of Proposition \ref{prop:interwining property local}, we deduce that, for all $n\in\N$, 
	\begin{align*}
		\int_\R\abso{\scal{e^{itH}W_+(H,H_0,I)u_n}{v}_\Hi}^2\mathrm{d}t&=\int_\R\abso{\scal{e^{itH_0}u_n}{W_+(H,H_0,I)^\star v}_\Hi}^2\mathrm{d}t\\
		&\leq c_{u_n}\norme{W_+(H,H_0,I)}_{\mathcal{B}(\Hi)}\norme{v}_\Hi^2. 
	\end{align*}
	Thus $(W_+(H,H_0,I)u_n)_{n\in\N}\subset\mathcal{M}(H)$, which implies \eqref{eq:proof kernel local wave operator recip}. Finally, using the spectral resolution formula \eqref{eq:resolution identity} together with \eqref{eq:a1}, \eqref{eq:a2} and \eqref{eq:proof kernel local wave operator recip} we obtain
	\begin{equation*}
		r(H)W_+(H,H_0,I)u=0.
	\end{equation*}
	Since the restriction $r(H)$ to $\Hi_\mathrm{ac}(H)$ is injective, it follows that $u\in\Ker(W_+(H,H_0,I)$. 
	
	The characterizations of the ranges follow from Propositions \ref{prop:adjoint_local} and \ref{prop: range and kernel of local wave operators bis}.
	%
	\end{proof}


Now we turn to the study of the kernels and ranges of $W_\pm(H_0,H,I)$.

\begin{proposition}\label{prop: range and kernel of local wave operators bis}
	Suppose that Hypotheses \ref{hyp:LAP H0}-\ref{hyp:Conjugate operator} hold. Let $I\subset\sigma_\mathrm{ess}(H)$ be a closed interval. Then 
	\begin{equation*}
		\Ker((h\mathds{1}_I)(H))=\Ker(W_\pm(H_0,H,I)) \quad \text{and}\quad \Ker((\bar{h}\mathds{1}_I)(H^\star))= \Ker((W_\pm(H_0,H^\star,I)).
	\end{equation*}
	Moreover we have 
	\begin{equation*}
		 \Ran(W_\pm(H_0,H,I))^\mathrm{cl} = \Ran(\mathds{1}_I(H_0)) \quad \text{and}\quad \Ran(W_\pm(H_0,H^\star,I))^\mathrm{cl} = \Ran(\mathds{1}_I(H_0)).
	\end{equation*}
\end{proposition}

\begin{proof}
	We show that $\Ker(W_+(H_0,H,I))=\Ker((h\mathds{1}_I)(H))$, the proof the other equalities characterizing the kernels are similar. Let $u\in\Ker(W_+(H_0,H,I))$. Then 
	\begin{equation*}
		\lim_{t\rightarrow \infty}\norme{e^{itH_0}e^{-itH}(h\mathds{1}_I)u}_\Hi=0
	\end{equation*}
	and therefore, by unitarity of $e^{itH_0}$,
	\begin{equation*}
			\lim_{t\rightarrow \infty}\norme{e^{-itH}(h\mathds{1}_I)u}_\Hi=0. 
	\end{equation*}
Hence $(h\mathds{1}_I)u\in\Hi_\mathrm{ads}(H)^+$ (see \eqref{eq:defHads}) which in turn yields $(h\mathds{1}_I)u\in\Hi_\mathrm{p}^+(H)$ by \eqref{eq:defHads2}. Thus we have 
\begin{equation*}
	(h\mathds{1}_I)(H)u=\Pi_\mathrm{p}^+(H)(h\mathds{1}_I)(H)u.
\end{equation*}
Using \eqref{eq:def regularized spectral projection} it is not difficult to verify that, for all $v\in\Hi$, 
\begin{equation*}
	\scal{\Pi_\mathrm{p}^+(H)(h\mathds{1}_I)(H)u}{v}_\Hi=0. 
\end{equation*}
Therefore $\Pi_\mathrm{p}^+(H)(h\mathds{1}_I)(H)u=0$ and hence $u\in \Ker((h\mathds{1}_I)(H))$. 

Conversely, if $u\in\Ker((h\mathds{1}_I)(H))$, then by definition of $W_+(H_0,H,I)$, $u\in\Ker(W_+(H_0,H,I))$. 

	The characterizations of the ranges follow from Propositions \ref{prop:adjoint_local} and \ref{prop:kernel and range local wave operator}.
\end{proof}


\section{Asymptotic completeness}\label{sec:Assymptotic completeness}

In this section, we study a notion of asymptotic completeness for the wave operator $W_\pm(H,H_0)$. 

\begin{definition}
We say that the wave operator $W_\pm(H,H_0)$ or $W_\pm(H^\star,H_0)$ is asymptotically complete if its range is closed.  
\end{definition}

This definition is motivated by Proposition \ref{prop:range_and_ker_of_wave_operator}. Indeed, we already know that, under our assumptions, the wave operators $W_\pm(H,H_0)$ are injective and that their ranges are dense in $\Hi_\mathrm{ac}(H)$. Hence if the wave operators are asymptotically complete, they are bijective. In the following we give conditions ensuring that asymptotic completeness holds. Note that $W_-(H,H_0)$ may be asymptotically complete while $W_+(H,H_0)$ is not. 

\subsection{Asymptotic completeness and spectral singularities}

Our first result shows that if $H$ does not have any spectral singularity, then the wave operators are asymptotically complete. A similar result has been proved in \cite{FaFr18_01} for dissipative operators.

We recall that if $H$ has no incoming/outgoing spectral singularity, then $W_\pm(H,H_0)$ and $W_\pm(H_0,H)$ reduce to 
\begin{equation*}
W_\pm(H,H_0)=\slim_{t\rightarrow\pm\infty}e^{itH}\Pi_\mathrm{ac}(H)e^{-itH_0}, \quad W_\pm(H_0,H):=\slim_{t\rightarrow\pm\infty}e^{itH_0}\Pi_\mathrm{ac}(H)e^{-itH}.
\end{equation*}

\begin{theorem}\label{thm:asymptotic completeness if absence of spectral singularities}
Suppose that Hypotheses \ref{hyp:LAP H0}-\ref{hyp:Conjugate operator} hold. Suppose that $H$ does not have any spectral singularities. Then $W_\pm(H,H_0)$ and $W_\pm(H^\star,H_0)$ are asymptotically complete. 
\end{theorem}

\begin{proof}
We prove that $W_+(H,H_0)$ is asymptotically complete. By proposition \ref{prop:range_and_ker_of_wave_operator}, we already know that $\Ran(W_+(H,H_0))\subset \Hi_\mathrm{ac}(H)$. Hence, to prove that the range of $W_+(H,H_0)$ is closed, it suffices to show that $\Hi_\mathrm{ac}(H)\subset\Ran(W_+(H,H_0))$. Let $u\in\Hi_\mathrm{ac}(H)$. Then 
\begin{align*}
u&=\Pi_\mathrm{ac}(H)\Pi_\mathrm{ac}(H)u=\Pi_\mathrm{ac}(H)e^{itH}e^{-itH_0}e^{itH_0}e^{-itH}\Pi_\mathrm{ac}(H)u.
\end{align*}
As $e^{itH}$ is uniformly bounded in $t\in\mathbb{R}$ by \eqref{eq:group of H uniformly bounded }, letting $t\rightarrow\infty$, we easily obtain 
\begin{equation*}
u=W_+(H,H_0)W_+(H_0,H)u.
\end{equation*}
Thus $u\in\Ran(W_+(H,H_0))$, which concludes the proof.
\end{proof}

Under the conditions of Theorem \ref{thm:asymptotic completeness if absence of spectral singularities}, we see in particular that $W_+(H,H_0)$ is a left inverse of $W_+(H_0,H)$ and, as $W_+(H,H_0)$ is invertible in $\mathcal{B}(\Hi,\Hi_\mathrm{ac}(H))$, we have
\begin{equation*}
W_+(H,H_0)^{-1}=W_+(H_0,H).
\end{equation*}
Note also that the intertwining properties show that 
\begin{equation*}
Hu=W_+(H,H_0)H_0W_+(H_0,H)u,\quad \forall u\in\mathcal{D}(H_0)\cap\Hi_\mathrm{ac}(H),
\end{equation*}
and 
\begin{equation*}
e^{itH}u=W_+(H,H_0)e^{itH_0}W_+(H_0,H)u,\quad \forall u\in\Hi_\mathrm{ac}(H),\quad\forall t\in \R. 
\end{equation*}
In other words, the restriction of $H$ to $\Hi_\mathrm{ac}(H)$ is similar to $H_0$.

The next result is a reformulation of Theorem \ref{mainthm:boundedness of solution of Schrodinger equation}.

\begin{theorem}\label{thm:boundedness of solution of Schrodinger equation}
Suppose that Hypotheses \ref{hyp:LAP H0}-\ref{hyp:Conjugate operator} hold. Suppose that $H$ does not have spectral singularities. Then there exist $m_1>0$, and $m_2>0$ such that 
\begin{equation*}
m_1\norme{u}_\Hi\leq \norme{e^{itH}u}_\Hi\leq m_2\norme{u}_\Hi, \quad \forall u\in\Hi_\mathrm{ac}(H),\quad \forall t \in \R. 
\end{equation*}
\end{theorem}

\begin{proof}
Let $u\in\Hi_\mathrm{ac}(H)$. Using the functional calculus \eqref{eq:functional_calculus_normal}, we deduce that there exists $m_2>0$ such that for all $t\in\R$, for all $v\in\Hi$, 
\begin{equation*}
\abso{\scal{v}{e^{itH}u}_\Hi}=\abso{\frac{1}{2\pi i}\lim_{\varepsilon\rightarrow 0^+} \int_{\sigma_\mathrm{ess}(H)}\scal{v}{e^{it\lambda}\left(\Res_H(\lambda+i\varepsilon)-\Res_H(\lambda-i\varepsilon)\right)u}_\Hi\mathrm{d}\lambda}\leq m_2\norme{u}_\Hi\norme{v}_\Hi.
\end{equation*}
Thus 
\begin{equation*}
\norme{e^{itH}u}_\Hi\leq m_2\norme{u}_\Hi.
\end{equation*}
For the other inequality, we use that for all $t\in\R$, 
\begin{equation*}
\norme{e^{itH}u}_\Hi=\norme{W_+(H,H_0)e^{itH_0}W_+(H_0,H)u}_\Hi.
\end{equation*}
By the remarks above, $W_+(H,H_0)$ and $W_+(H_0,H)$ are invertible in $\mathcal{B}(\Hi,\Hi_{\mathrm{ac}}(H))$ and $\mathcal{B}(\Hi_{\mathrm{ac}}(H),\Hi)$, respectively, and the evolution group generated by $H_0$ is unitary. Hence there exists $m_1>0$ such that 
\begin{equation*}
\norme{e^{itH}u}_\Hi\geq m_1\norme{u}_\Hi
\end{equation*}
This concludes the proof. 
\end{proof}

Note that the proof of the second inequality of the theorem above does not involve scattering theory but only functional calculus. 

\subsection{Completeness of the local wave operators} 

%

We recall that if $I$ is a closed interval contained in $\sigma_{\mathrm{ess}}(H)$ such that $H$ does not have any spectral singularity in $I$, the the local wave operators associated to $H$ and $H_0$ reduce to
\begin{equation*}
	W_\pm(H,H_0,I):=\slim_{t\rightarrow\pm \infty} e^{itH}\mathds{1}_I(H)e^{-itH_0}
\end{equation*}
The next theorem shows that in this case $W_\pm(H,H_0,I)$ are asymptotically complete in a natural sense.

\begin{theorem}
	Suppose that Hypotheses \ref{hyp:LAP H0}-\ref{hyp:Conjugate operator} hold. Let $I\subset\sigma_\mathrm{ess}(H)$ be a closed interval not containing any spectral singularity of $H$. Then $W_\pm(H,H_0,I)$ and $W_\pm(H^\star,H_0,I)$ are asymptotically complete in the sense that
	\begin{align*}
	\Ran(W_+(H,H_0,I))=\Ran(\mathds{1}_I(H)), \quad \Ran(W_+(H^\star,H_0,I))=\Ran(\mathds{1}_I(H^\star)).
	\end{align*} 
\end{theorem}

\begin{proof}
	We prove the equality for $W_+(H,H_0,I)$. The proof of the equalities for $W_-(H,H_0,I)$ and $W_\pm(H^\star,H_0,I)$ are similar. If $I$ does not have any spectral singularity, then with the interwining property, 
	\begin{equation}\label{eq:proof inversibility local wave operator}
		W_+(H,H_0,I)=\mathds{1}_I(H)W_+(H,H_0,I)=W_+(H,H_0,I)\mathds{1}_I(H_0).
	\end{equation}
	So by proposition \ref{prop:kernel and range local wave operator} and with \ref{eq:proof inversibility local wave operator} the restriction of $W_+(H,H_0,I)$ to $\Ran(\mathds{1}_I(H_0))$ is injective and $\Ran(W_+(H,H_0,I))\subset\Ran(\mathds{1}_I(H))$. We will show the reverse inclusion. Let $u\in\Ran(\mathds{1}_I(H))$, then 
	\begin{align*}
		u=\mathds{1}_I(H)u=e^{itH}\mathds{1}_I(H)e^{-itH_0}e^{itH_0}e^{-itH}\mathds{1}_I(H). 
	\end{align*} 
	As $t\mapsto e^{itH}\mathds{1}_I(H)e^{-itH_0}$, is uniformely bounded in $\mathcal{B}(\Hi)$, then 
	\begin{align*}
		u=e^{itH}\mathds{1}_I(H)e^{-itH_0}W_+(H_0,H,I)u+o(1),\quad (t\rightarrow\infty)\\
		=W_+(H,H_0,I)W_+(H_0,H,I)u
	\end{align*}
	This prove the reverse inclusion. 
\end{proof}

The previous result shows that if $H$ does not have any spectral singularity in $I$, then the inverses of $W_\pm(H,H_0,I)$ are $W_\pm(H_0,H,I)$. It should be noted that in order to construct the local wave operators on an interval $I$ not containing any spectral singularity, it suffices to have the existence of the spectral projection in $I$. Next it follows from Theorem \ref{thm:local asymptotic completeness } that they are invertible, and hence in particular injective with closed range. In the general case where $I$ contains spectral singularities, on the other hand, we relied on the whole spectral structure of $H$ in order to prove that $W_\pm(H,H_0,I)$ are injective (see Proposition \ref{prop:kernel and range local wave operator}).

\subsection{Remarks on the notion of asymptotic completeness}

Theorem \ref{thm:asymptotic completeness if absence of spectral singularities} shows that if $H$ does not have any spectral singularities then the wave operators $W_\pm(H,H_0)$ are asymptotically complete. We conjecture that if $H$ has a spectral singularity, then the regularized wave operators are \emph{not} asymptotically complete. We give in this section partial results supporting this conjecture.

First we consider regularized wave operators with an `over-regularizing' operator in the following sense. For the sake of simplicity, suppose that $H$ has only one incoming spectral singularity $\lambda_0$ of finite order $\nu_0$. Consider the regularized wave operator 
\begin{equation*}
	\widetilde{W}_+(H,H_0):=\slim_{t\rightarrow\infty}e^{itH}\Pi_\mathrm{ac}(H)(H-\lambda_0)^{\nu_0+1}\Res_H(z_0)^{\nu_0+1}e^{-itH_0}.
\end{equation*}
This definition should be compared to the Definition \ref{def:ref wave operators} where the regularizing operator is $(H-\lambda_0)^{\nu_0}\Res_H(z_0)^{\nu_0}$ instead of $(H-\lambda_0)^{\nu_0+1}\Res_H(z_0)^{\nu_0+1}$.
In particular we have that 
\begin{equation*}
	\widetilde{W}_+(H,H_0)=(H-\lambda_0)\Res_H(z_0)W_+(H,H_0),
\end{equation*}
where $W_+(H,H_0)$ is the regularized wave operator associated to $H$ and $H_0$ from Definition \ref{def:ref wave operators}. Under our assumptions, it follows from Proposition \ref{prop:wave operator exist} that $\widetilde{W}_+(H,H_0)$ exists, is injective and its range is dense in $\Hi_\mathrm{ac}(H)$. The next proposition shows that $\widetilde{W}(H,H_0)$ is not surjective. 

\begin{proposition}\label{prop:not asymptotical completeness too much regularized}
	Suppose that Hypotheses \ref{hyp:LAP H0}-\ref{hyp:Conjugate operator} hold. Then $\widetilde{W}(H,H_0)$ is not asymptotically complete. 
\end{proposition}

\begin{proof}
	Suppose that $\widetilde{W}(H,H_0)$ is asymptotically complete. Then its range is closed and hence there exists $c>0$ such that, for all $u\in\Hi$, 
	\begin{equation*}
		\norme{\widetilde{W}(H,H_0)u}_\Hi\geq c\norme{u}_\Hi.
	\end{equation*}	
	It then follows from \eqref{eq:group of H uniformly bounded } that there exists $c_1>0$ such that 
	\begin{equation*}
		c_1\lim_{t\rightarrow\infty}\norme{(H-\lambda_0)\Res_H(z_0)e^{-itH_0}u}_\Hi\geq c\norme{u}_\Hi.
	\end{equation*}
As $(H-\lambda_0)\Res_H(z_0)-(H_0-\lambda_0)\Res_{H_0}(z_0)$ is compact, we deduce that
\begin{equation*}
	c_1\lim_{t\rightarrow\infty}\norme{(H_0-\lambda_0)\Res_{H_0}(z_0)e^{-itH_0}u}_\Hi\geq c\norme{u}_\Hi.
\end{equation*}
Finally, using that $\lbrace e^{-itH_0}\rbrace_{t\in\R}$ is unitary and that $\Res_{H_0}(z_0)$ is bounded, this implies that there exists $c_2>0$ such that
\begin{equation*}
	c_2\norme{(H_0-\lambda_0)u}_\Hi\geq c\norme{u}_\Hi.
\end{equation*}
This is a contradiction because this would imply that $H_0-\lambda_0$ is bijective (since it is self-adjoint and injective with closed range), which is impossible since $\lambda_0$ is in the essential spectrum of $H_0$. 
\end{proof}

Our last result shows that if $H$ has a incoming spectral singularity and if $-iH$ is the generator of a strongly continuous one parameter semigroup which is uniformly bounded for positive times, then $W_+(H,H_0)$ is not asymptotically complete. It's especially the case if $\im(V)\leq 0$ because in this case $-iH$ become dissipative and thus generates a contraction group. 

\begin{proposition}
	Suppose that Hypotheses \ref{hyp:LAP H0}-\ref{hyp:Conjugate operator} hold. If $\lbrace e^{itH}\rbrace_{\pm t\ge0}$ is uniformly bounded and $H$ has a incoming/outgoing spectral singularity then $W_\pm(H,H_0)$ are not asymptotically complete. 
\end{proposition}

\begin{proof}
	We prove the proposition for $W_+(H,H_0)$. The proof for $W_-(H,H_0)$ is very similar. Suppose that $W_+(H,H_0)$ is asymptotically complete. Then there exists $c>0$ such that for all $u\in\Hi$, 
	\begin{equation*}
		\norme{W_+(H,H_0)u}_\Hi\geq c\norme{u}_\Hi.
	\end{equation*}
Using that $\lbrace e^{itH}\rbrace_{t\ge0}$ is uniformly bounded and arguing as in the proof of Proposition \ref{prop:not asymptotical completeness too much regularized}, we have that there exists $c_1>0$ such that, for all $u\in\Hi$,
\begin{equation*}
	\norme{r_-(H_0)u}_\Hi\geq c_1\norme{u}_\Hi.
\end{equation*}
Since $r_-(H_0)$ is self-adjoint and injective, this would imply that it is bijective. However $r_-(H_0)$ is not bijective by definition, because spectral singularities of $H$ are points of the essential spectrum of $H_0$. 
\end{proof}

\appendix

\section{Proofs of Propositions \ref{prop:existence strong limit wave operator general 1} and \ref{prop:interwining property}}\label{App:existence and property wave operators}

In this appendix, for completeness, we recall the proofs of some standard results about wave operators. 

\begin{proof}[proof of Proposition \ref{prop:existence strong limit wave operator general 1}]
	We suppose that there exist $c>0$ such that 
	\begin{equation}\label{eq:proof existence wave operator estimate H}
		\int_0^\infty \norme{Ce^{-itH^\star}A^\star u}^2_\Hi\mathrm{d}t\leq c\norme{u}^2_\Hi. 
	\end{equation}
	and we prove that the strong limit $\slim_{t\rightarrow\infty} Ae^{itH}e^{-itH_0}$  exist. The existence of the other strong limits are similar. Because of the uniform boundedness principle, it suffices to show that for all $u\in\Hi$, 
	\begin{equation*}
		\lim_{t\rightarrow\infty} Ae^{itH}e^{-itH_0}u
	\end{equation*}
	exists in $\Hi$. We claim that for all $t>0$, for all $u\in\Hi$, we have
	\begin{equation}\label{eq:Wave operator integral t}
		Ae^{itH}e^{- itH_0}u=Au + i\int_0^t Ae^{isH}CWCe^{-isH_0}u\mathrm{d}s.
	\end{equation}
	Indeed, let $u\in\mathcal{D}(H_0)$ and $t\in\R$, then as $\lbrace  e^{itH}\rbrace _{t\in\R}  $ preserves $\mathcal{D}(H_0)$, then 
	\begin{align*}
		Ae^{itH}e^{-itH_0}u&=Au-\int_0^t \partial_s\left(Ae^{isH}e^{-isH_0}u\right)\mathrm{d}s\\
		&=Au-\int_0^t Ae^{isH}(iH-iH_0)e^{-isH_0}u\mathrm{d}s\\
		&=Au-i\int_0^t Ae^{isH}Ve^{-isH_0}u\mathrm{d}s.
	\end{align*}
	As the equality 
	\begin{equation*}
		Ae^{itH}e^{-itH_0}u=Au- i\int_0^t Ae^{isH}Ve^{-isH_0}u\mathrm{d}s
	\end{equation*}
	makes sense if  $u\in\Hi$, using that $\mathcal{D}(H_0)$ is dense in $\Hi$ together with Lebesgue's dominated convergence theorem, we deduce that \eqref{eq:Wave operator integral t} holds.
	 
	To prove that the limit in the left hand side of \eqref{eq:Wave operator integral t} exists when $t\rightarrow \infty$, we show that the integral in the right hand side of \eqref{eq:Wave operator integral t} converges when $t\rightarrow\infty$, To do this, we  show that the integral satisfies a Cauchy criterion. As $C$ is relatively smooth with respect to $H_0$, it follows from \eqref{eq:Cr Kato smooth H star} that there exists $c_0>0$ such that, for all $x,y\in\R$, 
	\begin{align*}
		&\norme{\int_x^y Ae^{isH}CWCe^{-isH_0}u\mathrm{d}s}_{\Hi}\\
		&=\sup_{\substack{v\in \Hi\\\norme{v}_\Hi=1}}\abso{\int_x^y\scal{Ce^{-isH^\star}A^\star v}{WCe^{-isH_0}u}_\Hi\mathrm{d}s}\\
		&\leq \sup_{\substack{v\in \Hi\\\norme{v}_\Hi=1}}\int_x^y\abso{\scal{Ce^{-isH^\star}Av}{WCe^{-isH_0}u}_\Hi}\mathrm{d}s\\
		&\leq \sup_{\substack{v\in \Hi\\\norme{v}_\Hi=1}}\left(\int_x^y\norme{Ce^{- isH^\star}Av}_\Hi^2\mathrm{d}s\right)^{1/2}\left(\int_x^y\norme{WCe^{-isH_0}u}_\Hi^2\mathrm{d}s\right)^{1/2}\\
		&\leq c_0\norme{u}_\Hi,
	\end{align*}
	where we used \eqref{eq:proof existence wave operator estimate H} and \eqref{eq:Kato_smooth} in the last inequality. Since $\Hi$ is complete the map $s\mapsto Ae^{isH}Ve^{-isH_0}u$ belongs to $L^1([0,\infty),\Hi)$ and hence the limit of \eqref{eq:Wave operator integral t} exists when t$\rightarrow \infty$. 
\end{proof}

\begin{proof}[proof of Proposition \ref{prop:interwining property}]
	We show the proposition for $W_+(H,H_0)$, the proof is similar for the other wave operators. Let $u\in\Hi$ and $t\in\R$. Then 
	\begin{align*}
		e^{itH}W_+(H,H_0)u&=e^{itH}e^{isH}\Pi_\mathrm{ac}(H)r_-(H)e^{-isH_0}u+o(1),\quad (s\rightarrow \infty)\\
		&=e^{i(t+s)H}\Pi_\mathrm{ac}(H)r_-(H)e^{-i(t+s)H_0}e^{itH_0}u+o(1),\quad (s\rightarrow \infty)\\
		&=e^{irH}\Pi_\mathrm{ac}(H)r_-(H)e^{-irH_0}e^{itH_0}u+o(1),\quad (r\rightarrow\infty)\\
		&=W_+(H,H_0)e^{itH_0}u. 
	\end{align*}
	This shows \eqref{eq:interwining property group}. To prove that $W_+(H,H_0)$ preserves $\mathcal{D}(H_0)$, it suffices to use that 
	\begin{equation*}
		\mathcal{D}(H_0)=\left\lbrace u\in\Hi,\quad \lim_{t\rightarrow 0^+}\frac{1}{-it}\left(e^{-itH}-\Id\right)u ~\text{exists in} ~\Hi\right\rbrace.
	\end{equation*}
	together with \ref{eq:interwining property group}. To prove \eqref{eq:interwining property operator}, we use that 
	\begin{equation*}
		Hu=\lim_{t\rightarrow 0}\frac{1}{-it}\left(e^{-itH}-\Id\right)u
	\end{equation*}
	and \eqref{eq:interwining property group} again. 
\end{proof}

\textit{Acknowledgements:} I would like to warmly thank Jérémy Faupin for his advice, useful discussions and his encouragements.

\end{document}